%% file: pcdnn.tex
\documentclass[onefignum,onetabnum]{siamonline171218}

\input{ex_shared}

\usepackage{lipsum}
\usepackage{amsfonts}
\usepackage{graphicx}
\usepackage{epstopdf}
\usepackage{enumerate}
\usepackage{cite}

\usepackage{subcaption}

\usepackage{tabularx}
\usepackage{multirow}
\usepackage{caption}
\usepackage{microtype}

\usepackage{algorithm}
\usepackage{algpseudocode}
\usepackage{amsmath}
\DeclareMathOperator*{\argmin}{argmin}

\usepackage{mathtools}
\DeclarePairedDelimiter\ceil{\lceil}{\rceil}

\ifpdf
\hypersetup{
  pdftitle={
  Point Cloud Surface Parametrization using HAND and LEG: Hausdorff Approximation from Node-wise Distances and Localized Energy for Geometry
  },
  pdfauthor={
    K. H. Lai, and L. M. Lui
  }
}
\fi

\begin{document}

\maketitle

\begin{abstract}
Surface parametrization is a crucial part in various fields, having applications in computer graphic, medical imaging, scientific computing and computational engineering. The majority of surface parametrization approaches are performed on triangular meshes. 
On the contrary, the theories and methods of point cloud surface parametrization are less researched, despite its rising significance. 
In this work, we compute surface parametrization in an optimization approach using neural networks, with novel loss functions introduced without extrinsic information, together with theoretical analyses. 
Based on the theory, we develop an optimization algorithm to improve the parametrization quality. Using our methods, general open surfaces can be parametrized in either free-boundary manner or with arbitrary domain constraints. Landmark matching can also be enforced under our framework. Numerical experiments are conducted and presented, along with applications including surface reconstruction and boundary detection. 
\end{abstract}



\input{sections/sec1-introduction}
\input{sections/sec2-contribution}

\input{sections/sec3-previous_work}

\input{sections/sec4-math_background}

\input{sections/sec5-proposed_methods}

\input{sections/sec6-implementation}

\input{sections/sec7-experiment}

\input{sections/sec8-application}
\input{sections/sec9-conclusion}


\bibliographystyle{siamplain}
\bibliography{references}
\end{document}

%% file: ex_shared.tex
\usepackage{lipsum}
\usepackage{amsfonts}
\usepackage{graphicx}
\usepackage{epstopdf}
\ifpdf
  \DeclareGraphicsExtensions{.eps,.pdf,.png,.jpg}
\else
  \DeclareGraphicsExtensions{.eps}
\fi

\newsiamremark{remark}{Remark}
\newsiamremark{hypothesis}{Hypothesis}
\crefname{hypothesis}{Hypothesis}{Hypotheses}
\newsiamthm{claim}{Claim}

\headers{point cloud parametrization using HAND \& LEG}{K. H. Lai, L. M. Lui}

\title{Point Cloud Surface Parametrization using HAND and LEG: Hausdorff Approximation from Node-wise Distances and Localized Energy for Geometry\thanks{Submitted to the editors DATE.
\funding{The second author was supported by HKRGC GRF (Project ID: 14307621).
}}}

\author{Ka Ho Lai\thanks{Department of Mathematics, The Chinese University of Hong Kong
  (\email{khlai@math.cuhk.edu.hk}
  ).}
\and Lok Ming Lui\thanks{Department of Mathematics, The Chinese University of Hong Kong
  (\email{lmlui@math.cuhk.edu.hk}
  ).}
}

\usepackage{amsopn}

\makeatletter
\newcommand*{\addFileDependency}[1]{
  \typeout{(#1)}
  \@addtofilelist{#1}
  \IfFileExists{#1}{}{\typeout{No file #1.}}
}
\makeatother


%% file: sections/sec1-introduction.tex
\section{Introduction}
In recent years, 3D scanning technologies have seen significant advancements. High-quality point clouds can now be obtained through scanning and applied to different fields in academia and industry. Due to its significance, many researchers have conducted studies on point cloud denoising, registration, recognition, and surface reconstruction using various methods and approaches. 
In the tasks, surface parametrization methods are frequently used, especially in point cloud surface registration \cite{tam2012registration} and reconstruction \cite{lim2014surface, choi2016spherical}. By mapping complicated shapes to simple parameter domains, the difficulty required for the processing is greatly reduced. However, many of methods of point cloud parametrization are achieved through surface approximation techniques, including local mesh construction \cite{Choi2021PointCloud}, moving least square \cite{choi2016spherical, meng2016tempo, meng2018pcbc} or lattice approximation \cite{wu2020computing}, to construct the desired differential operators. 
A point cloud is merely a set of points, and does not intrinsically contain boundary or connectivity information. Different choices of the boundary or connectivity lead to different results, and inaccurate estimations can accumulate errors throughout the process. 
This challenge motivates our research toward developing methods to avoid the uncertainties due to connectivity and boundary determination. 

In domain-constrained surface parametrization, the surface boundary is typically matched to boundary of the prescribed parameter domain. Meanwhile, connectivity is employed to approximate differential operator for specific geometric purpose. 
In this paper, we formulate the task of point cloud surface parametrization as an optimization problem with novel objective functions to address the stated possible problems without boundary or connectivity. 
The first loss function, called Hausdorff Approximation from Node-wise Distances (HAND), is a smooth approximation to Hausdorff distance, which is designed to address the boundary fitting problem. It transforms the problem of matching surface boundary and domain boundary into matching surface mapping and parameter domain. 
Its minimization provides a soft constraint on the domain alignment. 
As a generalization to traditional point-wise correspondence methods, HAND allows landmark-target correspondence with uneven numbers of points, which can be further generalized to region-wise correspondence. 
On the other hand, the second loss function, called Localized Energy for Geometry (LEG), measures geometric distortion between point clouds and their mappings. Inspired by conformal mapping, minimizing this function ensures that the resulting parametrization preserves the essential properties of the original point cloud, without the needs of approximations on the surface geometry. 
Theoretical analyses on HAND and LEG are performed on explaining the underlying mathematical meanings. 
Based on the insight from the theorems proved, a customized is algorithm developed for this optimization problem. It divides the overall process into stages, in order to refine detail quality. Experiments for the proposed methods are conducted, and applications, such as surface reconstruction and boundary detection, are demonstrated.

The organization of the paper is as follows. 
In Section \ref{sec: contributions}, the contributions of our work are listed. 
In Section \ref{sec: previous works}, previous related works are reviewed. 
In Section \ref{sec: math background}, some of the related mathematical concepts are discussed. 
In Section \ref{sec: proposed methods}, the proposed loss functions for the optimization, with theoretical analysis, and also an optimization scheme, are introduced. 
In Section \ref{sec: implementation}, some implementation details and experiment settings are stated.
In Section \ref{sec: experiments}, 
experimental results for the parameterization of point clouds with varying geometry and topology are presented. 
In Section \ref{sec: application}, boundary detection and surface reconstruction are shown as applications to our methods. 
Finally, we conclude our work in Section \ref{sec: conclusions}. 

%% file: sections/sec2-contribution.tex
\section{Contributions}
\label{sec: contributions}
In summary, our work has the following contributions:
\begin{enumerate}
    \item 
    We introduced Hausdorff Approximation from Node-wise Distances (HAND), enabling domain-constrained surface parametrization, without the presence of boundary information. Also, HAND can be applied in landmark matching. 

    \item 
    Inspired by conformal geometry, we introduced Localized Energy for Geometry (LEG), which can be computed directly on point clouds, without the approximation of differential operators on the surface. 

    \item 
    Theoretical analyses of the two loss functions are presented, relating them to the ground-truth Hausdorff distance and angle distortion, thus offering geometric insights into our methods.

    \item
    Building on the theoretical analyses of the minimization objectives, we designed an optimization algorithm that enhances the performance of minimizing the stated loss functions and improves the quality of the resulting mappings.

    \item 
    Our method is capable of parameterizing surfaces of any geometry and topology onto user-specified parameter domains of arbitrary geometry and topology. 

\end{enumerate}

%% file: sections/sec3-previous_work.tex
\section{Previous works}
\label{sec: previous works}
An isometry preserves surface metrics, and hence also angle and area, is an ideal parametrization. However, its existence is not guaranteed in general. 
Many researches then turn their focus on preserving particular geometric quantities or minimizing distortions. 

Conformal and quasiconformal mappings are commonly used in surface parametrization. Conformal maps preserve local geometry, while quasiconformal maps allow angle distortion, provides much more flexibility conformal maps. 
Gu et al. \cite{gu2004genus, gu2003global} proposed an algorithm to compute global conformal surface parametrization. 
Levy et al. \cite{levy2023least} introduced least square conformal energy for conformal maps. 
Wang et al. used \cite{wang2005brain} holomorphic 1-forms in computing the parametrization for brain surfaces. 
Lui et al. \cite{lui2010shape, Lui2012} developed a computational method of Beltrami holomorphic flow for surface registration. 
Zeng et al. \cite{Zeng2012} computed quasiconformal map using auxiliary metric method. 
More computational methods and theories about conformal maps can be found in the book written by Gu et al. \cite{CCG_2007}. 
Choi et al. \cite{choi2022recent} also review many quasiconformal methods.

Triangular meshes are commonly used in existing methods. In contrast, parameterizing point clouds is more challenging due to the lack of connectivity, resulting in fewer studies addressing this issue \cite{zwicker2004meshing, tewari2006meshing, zhang2010mesh}. Scientific computing on point clouds often rely on approximations of differential operators \cite{belkin2008towards, belkin2009constructing, liang2012geometric, liang2013solving}. 
For surface parametrization, methods of moving least square \cite{lange2005anisotropic}, are commonly employed for surface estimations. Using methods of moving least square, Choi et al. \cite{choi2016spherical} approximated Laplace–Beltrami operator to compute conformal maps, and Meng et al. \cite{meng2018pcbc} constructed point cloud Beltrami coefficient. Using defined point cloud Beltrami coefficient, Meng et al. \cite{meng2016tempo} computed point cloud Teichmüller mappings. Apart from moving least square, Choi et al. \cite{Choi2021PointCloud} constructed local meshes to approximate Laplace–Beltrami operator for computing free boundary conformal maps. 
Wu et al. \cite{wu2020computing} computed harmonic and conformal parametrization by lattice approximation on point clouds.  

Artificial neural networks \cite{Goodfellow-et-al-2016, bronstein2017geometric} are widely used in geometry processing. 
The continuity nature allows storing a continuous surface with infinite resolution using only a finite number of parameters. As a result, many researches on surface parametrization or surface representation are done using neural networks. 
Groueix et al. \cite{groueix2018papier} introduced AtlasNet to represent a 3D shape as a collection of parametric surface elements. 
Yang et al. \cite{yang2018foldingnet} used an autoencoder structure, namely FoldingNet, which deforms a 2D domain to a surface of a 3D point cloud object. 
Groueix et al. \cite{groueix20183d} introduced shape deformation networks for matching deformable shapes. 
Mescheder et al. \cite{mescheder2019occupancy} proposed occupancy networks, which represents 3D shapes by decision boundary of a classifier. 
Morreale et al. \cite{morreale2021neural} proposed neural surface map as a neural network representation to surface, and later an improvement using convolution \cite{morreale2022neural}.

%% file: sections/sec4-math_background.tex
\section{Mathematical Background}
\label{sec: math background}

\subsection{Conformal Mappings}

Denote two Riemann surfaces $\mathcal{M}, \mathcal{N}$, and their corresponding metrics $ds^2_\mathcal{M} = \sum_{i, j} g_{ij} dx^i dx^j$ and $ds^2_\mathcal{N} = \sum_{i, j} h_{ij} dx^i dx^j$. Considering a smooth mapping $f : \mathcal{M} \to \mathcal{N}$, we can define pull back metric induced by $f$. 

\begin{definition}
    \label{def: pull back metric}
    The pull back metric $f^* ds^2_\mathcal{N}$, induced by $f$ and $ds^2_\mathcal{N}$, is the metric defined on $\mathcal{M}$ as 
    \begin{equation}
        f^* ds^2_\mathcal{N} 
        = 
        \sum_{m, n} 
        \left( 
            \sum_{i, j}
            h_{ij}(
                f(x^1, x^2)
            )
            \frac{\partial f^m}{\partial x^i}
            \frac{\partial f^n}{\partial x^j}
        \right)
        dx^m dx^n. 
    \end{equation}
\end{definition}

Conformal maps between $\mathcal{M}$ and $\mathcal{N}$ is defined using pull back metric.
\begin{definition}
    \label{def: conformal map}
    A map $f$ is said to be a conformal map if there exists a positive scalar function $\lambda$ on $\mathcal{M}$ such that
    \begin{equation}
        f^* ds^2_\mathcal{N} = \lambda ds^2_\mathcal{M}.
    \end{equation}
    The function $\lambda$ is called conformal factor. 
\end{definition}

Intuitively, conformal map $f$ scales infinitesimal region by a factor $\lambda$, with angle preserved. 
More about conformal mappings could be found in \cite{schoen1994lectures}.

\subsection{Hausdorff Distance}
The Hausdorff distance is a measure of differences between subsets in a metric space. 
Let $(M, d)$ be a metric space and $X, Y \subseteq M$ be non-empty subsets of $M$. 

\begin{definition}
    \label{def: hausdorff distance}
    The Hausdorff distance between $X, Y$ is defined as
    \begin{equation}
    \begin{split}
        d_H(X, Y) 
        &= \max \biggl \{ \sup_{x \in X} \inf_{y \in Y} d(x, y) , \sup_{y \in Y} \inf_{x \in X} d(x, y) \biggr \}.
    \end{split}
    \end{equation}
\end{definition}

Write $d'_H(X, Y) = \sup_{x \in X} \inf_{y \in Y} d(x, y) + \sup_{y \in Y} \inf_{x \in X} d(x, y)$. 
Then we have $d_H(X, Y) = 0 \iff d'_H(X, Y) = 0 \iff \overline{X} = \overline{Y}$. By minimizing the Hausdorff distance $d_H$, or $d'_H$, a moving shape can be mapped to a target shape. However, the Hausdorff distance is not differentiable, and the true values are difficult to compute in general. 

Also, Hausdorff distance satisfies a triangle inequality:
\begin{lemma}
\label{lem: hausdorff triangle inequality}
    For non-empty subsets $A, B, C \subseteq M$, 
    \begin{equation}
        d_H(A, B) \le d_H(A, C) + d_H(C, B).
    \end{equation}
\end{lemma}

\subsection{Boltzmann Operator}
The Boltzmann Operator, introduced by Asadi et al. in \cite{asadi2017alternative}, is a smooth approximation for $\max$, $\min$ functions. 
\begin{definition}
    \label{def: boltzmann operator}
    The Boltzmann Operator is defined as:
    \begin{equation}
    B_\alpha(x_1, x_2, \dots, x_n) = \frac{\sum^n_{i=1} x_i e^{\alpha x_i} }{\sum^n_{i=1} e^{\alpha x_i}}. 
    \end{equation}
\end{definition}
Since the order of arguments is not important, we may write $B_\alpha( \Vec{x} )$ or $B_\alpha( \{ x_i \}_{i=1}^n )$ in the later context. 
An alternative formulation of Boltzmann Operator can be written as
\begin{equation}
    B_\alpha( \Vec{x} ) = \Vec{x} \cdot \text{Softmax} (\alpha \Vec{x}). 
\end{equation}
Boltzmann Operator is differentiable with gradient given by:
\begin{equation}
\label{eqn: boltzmann derivative}
\nabla B_\alpha (\Vec{x}) = \text{Softmax} (\alpha \Vec{x}) \odot [\Vec{\textbf{1}} + \alpha (\Vec{x} - B_\alpha(\Vec{x}) \Vec{\textbf{1}})], 
\end{equation}
where $\Vec{\textbf{1}} = (1, 1, \dots, 1) \in \mathbb{R}^n$. 

We can observe that $\min(\Vec{x}) \le B_\alpha(\Vec{x}) \le \max(\Vec{x})$.
If all values are the same, the Boltzmann Operator gets the exact maximum and minimum values.

%% file: sections/sec5-proposed_methods.tex
\section{Proposed Methods}
\label{sec: proposed methods}
Given a point cloud sampled from a surface embedded in $\mathbb{R}^3$, denoted by $\mathcal{X} = \{ x_i \}_{i=1}^N \subseteq \mathcal{M} \subseteq \mathbb{R}^3$, our goal is to find a parametrization $f: \mathcal{M} \to \Omega \subseteq \mathbb{R}^2$ where $\Omega$ is a pre-determined parameter domain. In this section, the problem is formulated as an optimization problem and the loss functions to be minimized are introduced. A deep learning algorithm is also designed and discussed for the optimization. The optimal mapping $f^*$ is obtained such that $d_H( f^* ( \mathcal{M} ), \Omega)$ and a local geometric distortion are small. 

\subsection{Hausdorff Approximation from Node-wise Distances}
In the past, similar tasks were done in a free boundary manner \cite{Choi2021PointCloud} or assuming boundary information is provided for fixed boundary purposes \cite{meng2016tempo}. 
In our work, provided only point cloud coordinates information without boundary information, a soft constraint on the parameter domain is given by minimizing an approximation to the Hausdorff distance. 
Karimi et al. \cite{karimi2019reducing} also made estimations of Hausdorff distance using distance transform, morphological erosion or convolutions with circular kernels for medical image segmentation. The methods are performed on digital images with regular grid structures, and they are difficult to be generalized to unstructured point cloud. 

Given two finite point sets $\mathcal{Y} = \{ y_i \}_{i=1}^N, \mathcal{W} = \{ w_k \}_{k=1}^M$ in the $n$-Euclidean space. Denote their pairwise Euclidean distance by $d^i_k = \| y_i - w_k \|$. The Hausdorff distance between $\mathcal{Y}$ and $\mathcal{W}$ is given by:
\begin{equation}
    d_H(\mathcal{Y}, \mathcal{W}) 
    = \max \left\{ 
        \max_{1 \le i \le N} \left(
            \min_{1 \le k \le M} \left(
                d^i_k
            \right)
        \right) 
        , 
        \max_{1 \le k \le M} \left(
            \min_{1 \le i \le N} \left(
                d^i_k
            \right)
        \right)
    \right\},
\end{equation}
and a modified Hausdorff distance is defined as: 
\begin{equation}
\label{def: d'_H}
    d'_H(\mathcal{Y}, \mathcal{W}) 
    = 
        \max_{1 \le i \le N} \left(
            \min_{1 \le k \le M} \left(
                d^i_k
            \right)
        \right) 
        + 
        \max_{1 \le k \le M} \left(
            \min_{1 \le i \le N} \left(
                d^i_k
            \right)
        \right). 
\end{equation}

Choosing $\alpha > 0$, a smooth approximation for $d'_H$ in (\ref{def: d'_H}), namely the Hausdorff Approximation from Node-wise Distances (HAND), is constructed using the Boltzmann operator in Definition \ref{def: boltzmann operator}, as below:
\begin{equation}
\label{def: DDHD}
\begin{split}
    H_\alpha(\mathcal{Y}, \mathcal{W}) 
    &= B_\alpha \Bigl( 
        \left\{ 
        B_{-\alpha} \left(
            \left\{
                d^i_k
            \right\}_{k = 1}^M
        \right) 
        \right\}_{i = 1}^N
    \Bigr) 
    + 
    B_\alpha \left(
        \left\{
        B_{-\alpha} \left(
            \left\{ 
                d^i_k 
            \right\}_{i = 1}^N
        \right) 
        \right\}_{k = 1}^M
    \right).
\end{split} 
\end{equation}

An error bound of approximation for Boltzmann operator is needed before further analysis. 
\begin{theorem}
\label{thm: boltzmann operator error bound}
For a set of descending number, $x_1 = x_2 = \dots = x_m > x_{m+1} \ge \dots \ge x_{n-l} > x_{n-l+1} = x_{n-l+2} = \dots = x_n$, write $\Vec{x} = (x_1, x_2, \dots, x_n)$, $x_{max} = x_1, x_{min} = x_n$. 
Then for $\alpha > 0$, 
\begin{equation}
| B_\alpha(\Vec{x}) - x_{max} | \le \frac{n}{m} e^{- \alpha (x_{max} - x_{m+1}) } \| \Vec{x} \| 
\end{equation} 
and
\begin{equation}
| B_{-\alpha}(\Vec{x}) - x_{min} | \le \frac{n}{l} e^{- \alpha (x_{n-l} - x_{min}) } \| \Vec{x} \|.
\end{equation}

\end{theorem}

\begin{proof}
For the first error bound approximating $x_{max}$, let $\alpha > 0$, $\Vec{v}_m \in \mathbb{R}^n$ such that $(\Vec{v}_m)_i = 1/m$ for $1 \le i \le m$ and equals $0$ for the rest, 
\begin{equation}
    \left| B_\alpha(\Vec{x}) - x_{max} \right| = \left| \Vec{x} \cdot \left( \text{Softmax}(\alpha \Vec{x}) - \Vec{v}_m \right) \right| \le \| \Vec{x} \| \|\text{Softmax}(\alpha \Vec{x}) - \Vec{v}_m \|
\end{equation}
where for $1 \le i \le m$, $m+1 \le j \le n$, 
\begin{equation}
    \begin{split}
        \left| (\text{Softmax}(\alpha \Vec{x}) - \Vec{v}_m)_i \right| 
        &= \frac{ \sum^n_{k=m+1} e^{\alpha(x_k - x_{max}) } }{ m (m + \sum^n_{k=m+1} e^{\alpha(x_k - x_{max})} )} \\
        &\le \frac{n - m}{m^2} e^{ - \alpha(x_{max} - x_{m+1}) },
    \end{split}
\end{equation}
and
\begin{equation}
    \begin{split}
        \left| (\text{Softmax}(\alpha \Vec{x}) - \Vec{v}_m)_j \right| 
        &= \frac{e^{\alpha (x_j - x_{max}) }}{ m + \sum^n_{k=m+1} e^{\alpha (x_k - x_{max}) } } \\
        &\le \frac{1}{m} e^{-\alpha (x_{max} - x_{m+1})}.
    \end{split}
    \end{equation}
    So, we have
    \begin{equation}
        \|\text{Softmax}(\alpha \Vec{x}) - \Vec{v}_m \| \le \frac{n}{m} e^{-\alpha (x_{max} - x_{m+1})},
    \end{equation}
    and hence
    \begin{equation}
     \left| B_\alpha(\Vec{x}) - x_{max} \right| \le \frac{n}{m} e^{-\alpha (x_{max} - x_{m+1})} \| \Vec{x} \|. 
    \end{equation}
    
    The second error bound comes from the first. Since 
    \begin{equation}
    B_{ - \alpha}(\Vec{x}) = -B_{\alpha}(-\Vec{x}),
    \end{equation}
    we have
    \begin{equation}
    | B_{ - \alpha}(\Vec{x}) - x_{min}| = |-B_{\alpha}(-\Vec{x}) + \max(-\Vec{x})| \le \frac{n}{l} e^{ - \alpha (x_{n-l} - x_{min})} \| \Vec{x} \|.
    \end{equation}
\end{proof}


Let $\mathcal{X}$ be a finite point cloud sampled from surface $\mathcal{M}$, and $\mathcal{W}$ is a finite sample on parameter domain $\Omega \subseteq \mathbb{R}^2$. The mapping $f: \mathcal{M} \to \mathbb{R}^2$ maps the point cloud $\mathcal{X}$ to $\mathcal{Y} = f(\mathcal{X}) \subseteq \mathbb{R}^2$. 
Write $h^\Omega_\mathcal{W} = d_H(\Omega, \mathcal{W})$, which is the Hausdorff distance between parameter domain $\Omega$ and a finite sample $\mathcal{W}$. 
This number measures the quality of the sampling inside the parameter domain. 
A denser point set is a better representation of the parameter domain, and the quantity is smaller. 
Similarly, write $h^{f(\mathcal{M})}_\mathcal{Y} = d_H(f(\mathcal{M}), \mathcal{Y})$, which is the Hausdorff distance between surface mapping $f(\mathcal{M})$ and the point cloud mapping $\mathcal{Y} = f(\mathcal{X})$. A more uniform and finer sampling on the same surface gives a smaller number in $h^{f(\mathcal{M})}_\mathcal{Y}$.

The theorem below relates $d_H( f( \mathcal{M} ), \Omega )$ with $H_\alpha(\mathcal{Y}, \mathcal{W})$. 

\begin{theorem}
\label{thm: error bound for approximation of Hausdorff Distance}
Suppose $ f(\mathcal{M}) , \Omega $ are bounded. 
Then 
\begin{equation}
\label{ineq: d_H < H_a}
    d_H( f( \mathcal{M} ), \Omega ) \le h^{f(\mathcal{M})}_\mathcal{Y} + h^\Omega_\mathcal{W} + H_\alpha(\mathcal{Y}, \mathcal{W}) + O(e^{-\alpha C})
    \quad \text{ as } \alpha \to \infty, 
\end{equation}
where $C$ is a constant and independent of $\alpha$.
\end{theorem}


\begin{proof}
\label{proof: error bound for approximation of Hausdorff Distance}
    For $1\le i \le N$ and $1 \le k \le M$, define the following notations:
    \begin{equation}
        \Vec{d^i} = (d^i_1, d^i_2, \cdots , d^i_M), \Vec{d_k} = (d^1_k, d^2_k, \cdots , d^N_k)
    \end{equation}
    represent the pairwise distances between two point sets, 
    \begin{equation}
    \begin{split}
        \Vec{b_{-\alpha}} &= ( B_{-\alpha}(\Vec{d_1}), B_{-\alpha}(\Vec{d_2}), \cdots , B_{-\alpha}(\Vec{d_M}) ), \\
        \Vec{b^{-\alpha}} &= ( B_{-\alpha}(\Vec{d^1}), B_{-\alpha}(\Vec{d^2}), \cdots , B_{-\alpha}(\Vec{d^N}) )  
    \end{split}
    \end{equation}
    are approximations using Boltzmann operator to the following: 
    \begin{equation}
    \begin{split}
        \Vec{d}^{\min} &= (\min (\Vec{d^1}), \min (\Vec{d^2}), \cdots, \min (\Vec{d^N}) ), \\
        \Vec{d}_{\min} &= (\min (\Vec{d_1}), \min (\Vec{d_2}), \cdots, \min (\Vec{d_M}) ).
    \end{split}
    \end{equation}

    By Lemma \ref{lem: hausdorff triangle inequality}, 
    \begin{equation}
    \label{ineq: error bound for d_H - H_a}
    \begin{split}
        &\quad d_H( f( \mathcal{M} ), \Omega ) - H_\alpha(\mathcal{Y}, \mathcal{W}) \\
        &\le h^{f(\mathcal{M})}_\mathcal{Y} + h^\Omega_\mathcal{W} + d'_H(\mathcal{Y}, \mathcal{W}) 
        - H_\alpha(\mathcal{Y}, \mathcal{W})
        \\
        &\le h^{f(\mathcal{M})}_\mathcal{Y} + h^\Omega_\mathcal{W} 
        + | B_\alpha ( 
        \Vec{b^{-\alpha}}
        )
        - 
        \max(\Vec{d}^{\min})
        | 
         + | B_\alpha (
        \Vec{b_{-\alpha}}
        ) 
        - 
        \max(\Vec{d}_{\min})
        |.
    \end{split}
    \end{equation}
    
    Denote $\text{secmax}(\Vec{x})$ to be the second largest distinct value in $\Vec{x}$, $\text{secmin}(\Vec{x})$ to be the second smallest distinct value in $\Vec{x}$, $s_{max}(\Vec{x}) = \max(\Vec{x}) - \text{secmax}(\Vec{x})$, $s_{min}(\Vec{x}) = \text{secmin}(\Vec{x}) - \min(\Vec{x})$. These functions are used in the exponent in the error bound provided by Theorem \ref{thm: boltzmann operator error bound}. If all values in $\Vec{x}$ are the same, the approximation will be exact and no error. 
    In that case, $\text{secmax}(\Vec{x})$ and $\text{secmin}(\Vec{x})$ can be redefined to be arbitrary small and large respectively. 
    
    Consider the term $| B_\alpha (\Vec{b^{-\alpha}}) - \max(\Vec{d}^{\min}) |$ in inequality (\ref{ineq: error bound for d_H - H_a}). Let $\iota$ be the set of indices that $\Vec{d}^{\min}$ attains its maximum, $\Vec{p} \in \mathbb{R}^N$ such that $(\Vec{p})_i = 1 / | \iota | $ if $i \in \iota$ and equals $0$ for the rest.
    Then: 
    \begin{equation}
    \label{ineq: B_alpha (b^-alpha) - max(d^min)}
    \begin{split}
        &\quad | B_\alpha (\Vec{b^{-\alpha}}) - \max(\Vec{d}^{\min}) | \\
        &\le | B_\alpha (\Vec{b^{-\alpha}}) - \Vec{p} \cdot \Vec{b^{-\alpha}} | + | \Vec{p} \cdot \Vec{b^{-\alpha}} -  \max(\Vec{d}^{\min}) | \\
        &\le \| \text{Softmax} (\alpha \Vec{b^{-\alpha}}) - \Vec{p} \| \| \Vec{b^{-\alpha}} \| + | \Vec{p} \cdot (\Vec{b^{-\alpha}} -  \Vec{d}^{\min}) |.
    \end{split}
    \end{equation}
    
    If $\max(\Vec{d}^{\min}) \ne \min(\Vec{d}^{min})$, which means that there are some indices not in $\iota$ , by Theorem \ref{thm: boltzmann operator error bound}, for $i \in \iota, j \notin \iota$, $B_{-\alpha}(\Vec{d^i}) - B_{-\alpha}(\Vec{d^j}) \ge \min(\Vec{d^i}) - \min(\Vec{d^j}) - M e^{-\alpha s_{min} (\Vec{d^j})} \|\Vec{d^j}\| \ge s_{max} (\Vec{d}^{\min}) - M e^{-\alpha s_{min}(\Vec{d^j})} \|\Vec{d^j}\| $.
    
    Take $\alpha$ to be large enough such that 
    \begin{equation}
    \label{ineq: bound for large alpha}
    B_{-\alpha}(\Vec{d^i}) - B_{-\alpha}(\Vec{d^j}) > s_{max} (\Vec{d}^{\min}) / 2, 
    \end{equation}
    for any $i \in \iota, j \notin \iota$. 
    
    First estimate the leftmost term, for $i \in \iota$, 
    \begin{equation}
    \label{ineq: (softmax - p)_i}
    \begin{split}
        &\quad
        \left| (\text{Softmax} (\alpha \Vec{b^{-\alpha}}) - \Vec{p})_i \right| 
        \\ &= 
        \left|
        \frac{
            |\iota| e^{\alpha B_{-\alpha} (\Vec{d^i}) } - \sum^N_{k=1}e^{\alpha B_{-\alpha} (\Vec{d^k})}
        }{
            |\iota| \sum^N_{k=1}e^{\alpha B_{-\alpha} (\Vec{d^k})}
        }
        \right| \\
        &\le
        \frac{
        \left|
            \sum_{k \in \iota} \left(
                e^{\alpha B_{-\alpha} (\Vec{d^i}) } - e^{\alpha B_{-\alpha} (\Vec{d^k})}
            \right)
        \right| 
        +
        \left|
            \sum_{k \notin \iota}
                e^{\alpha B_{-\alpha} (\Vec{d^k})}
        \right| 
        }{
            |\iota| \sum_{k \in \iota}e^{\alpha B_{-\alpha} (\Vec{d^k})}
        } 
    \end{split}
    \end{equation}
    For the left term of the numerator, 
    \begin{equation}
    \label{ineq: numerator left}
    \begin{split}
        \left|
            \sum_{k \in \iota} \left(
                e^{\alpha B_{-\alpha} (\Vec{d^i}) } - e^{\alpha B_{-\alpha} (\Vec{d^k})}
            \right)
        \right| 
        \le
            \sum_{k \in \iota} 
            \left|
                1 - e^{\alpha (B_{-\alpha} (\Vec{d^k}) - B_{-\alpha} (\Vec{d^i}))}
            \right|
        e^{ \alpha B_{-\alpha} ( \Vec{d^i})}. 
    \end{split}
    \end{equation}
    Using inequality (\ref{ineq: bound for large alpha}), the second term becomes
    \begin{equation}
    \label{ineq: numerator right}
    \begin{split}
        \left|
            \sum_{k \notin \iota}
                e^{\alpha B_{-\alpha} (\Vec{d^k})}
        \right| 
        \le 
        (N - | \iota |)
        e^{- \alpha s_{max}(\Vec{d}^{\min}) / 2} 
        e^{ \alpha B_{-\alpha} ( \Vec{d^i})}.
    \end{split}
    \end{equation}
    The denominator part becomes
    \begin{equation}
    \label{ineq: denominator}
    \begin{split}
        \left(
        |\iota| \sum_{k \in \iota}e^{\alpha B_{-\alpha} (\Vec{d^k})}
        \right)^{-1}
        \le
        \frac{
            1
        }{
            |\iota|^2
        } \max_{k \in \iota}e^{-\alpha B_{-\alpha} (\Vec{d^k})
        } 
        ,
    \end{split}
    \end{equation}
    and combining all (\ref{ineq: numerator left}), (\ref{ineq: numerator right}), (\ref{ineq: denominator}), 
    \begin{equation}
    \label{ineq: estimate i-th (softmax - p)}
    \begin{split}
        &\quad
        \left| (\text{Softmax} (\alpha \Vec{b^{-\alpha}}) - \Vec{p})_i \right| 
        \\ &\le 
        \frac{
            1
        }{
            |\iota|^2
        } \max_{k \in \iota}e^{-\alpha (B_{-\alpha} (\Vec{d^k}) - B_{-\alpha} (\Vec{d^i}))} 
        \\ &\quad
        \left(
        \sum_{k \in \iota} \left|
            1 - e^{\alpha (B_{-\alpha} (\Vec{d^k}) - B_{-\alpha} (\Vec{d^i}))}
        \right|
        + 
        (N - | \iota |)
        e^{- \alpha s_{max}(\Vec{d}^{\min}) / 2} 
        \right)
    \end{split}
    \end{equation}
    Note for $i, k \in \iota$, by Theorem \ref{thm: boltzmann operator error bound}, 
    \begin{equation}
        -Me^{-\alpha s_{min} (\Vec{d^i})} \|\Vec{d^i}\| \le B_{-\alpha}(\Vec{d^k}) - B_{-\alpha}(\Vec{d^i}) \le  Me^{-\alpha s_{min} (\Vec{d^k})} \|\Vec{d^k}\|
    \end{equation}
    and 
    \begin{equation}
    \begin{split}
        1 - e^{\alpha M \|\Vec{d^k}\| e^{-\alpha s_{min} (\Vec{d^k})}}
        \le
            1 - e^{\alpha (B_{-\alpha} (\Vec{d^k}) - B_{-\alpha} (\Vec{d^i}))}
        \le
        1 - e^{-\alpha M \|\Vec{d^i}\| e^{-\alpha s_{min} (\Vec{d^i})}}. 
    \end{split}
    \end{equation}
    Notice that by a easy check, for $0 < c < c_2$, 
        \begin{equation}
        e^{c_1 \alpha e^{-c_2 \alpha}} - 1 = O(e^{-c \alpha}) \quad \text{ as } \alpha \to \infty. 
    \end{equation}
    Then
    \begin{equation}
    \label{ineq: c^i_1}
        \max_{k \in \iota}e^{-\alpha (B_{-\alpha} (\Vec{d^k}) - B_{-\alpha} (\Vec{d^i}))} \le c_1^i
    \end{equation}
    and for all $k \in \iota$, 
    \begin{equation}
    \label{ineq: c^i_2, c^i_3}
        \left|
            1 - e^{\alpha (B_{-\alpha} (\Vec{d^k}) - B_{-\alpha} (\Vec{d^i}))}
        \right| \le c_2^i e^{-c_3^i \alpha}
    \end{equation}
    for some constant $c_1^i > 1, c_2^i, c_3^i > 0$. 
    Combining (\ref{ineq: c^i_1}) and (\ref{ineq: c^i_2, c^i_3}) into (\ref{ineq: estimate i-th (softmax - p)}), 
    \begin{equation}
        \left| (\text{Softmax} (\alpha \Vec{b^{-\alpha}}) - \Vec{p})_i \right| 
        \le c_1^i \frac{N - | \iota |}{|\iota|^2} e^{- \alpha s_{max}(\Vec{d}^{\min}) / 2} 
        + 
        \frac{c_1^i c_2^i}{| \iota |} e^{-c_3^i \alpha} 
    \end{equation}
    
    For $j \notin \iota$, 
    \begin{equation}
    \begin{split}
        \left|
        (\text{Softmax} (\alpha \Vec{b^{-\alpha}}) - \Vec{p})_j
        \right|
        &= 
        \frac{
            1
        }{
            \sum^N_{k=1} e^{\alpha (B_{-\alpha} (\Vec{d^k}) - B_{-\alpha} (\Vec{d^j}))}
        } \\
        &\le
        \frac{
            1
        }{
            \sum_{k \in \iota} e^{\alpha (B_{-\alpha} (\Vec{d^k}) - B_{-\alpha} (\Vec{d^j}))}
        } \\
        &\le 
        \frac{1}{|\iota|} e^{- \alpha s_{max}(\Vec{d}^{\min}) / 2}
        .
    \end{split}
    \end{equation}
    If on the other hand $\max(\Vec{d}^{\min}) = \min(\Vec{d}^{\min})$, which means $\iota = \{ 1, \dots N \}$, by defining $s_{max}(\Vec{d}^{\min})$ to be an arbitrary positive number and an approach similar to (\ref{ineq: (softmax - p)_i}), (\ref{ineq: numerator left}), 
    (\ref{ineq: denominator}), (\ref{ineq: estimate i-th (softmax - p)}), we have $ |\iota| = N $ and
    \begin{equation}
    \begin{split}
        &\quad 
        \left| (\text{Softmax} (\alpha \Vec{b^{-\alpha}}) - \Vec{p})_i \right| \\
        &\le 
        \frac{
            1
        }{
            N^2
        } \max_{1 \le k \le N}e^{-\alpha (B_{-\alpha} (\Vec{d^k}) - B_{-\alpha} (\Vec{d^i}))} 
        \sum_{k = 1}^N \left|
            1 - e^{\alpha (B_{-\alpha} (\Vec{d^k}) - B_{-\alpha} (\Vec{d^i}))}
        \right| \\
        &\le 
        c_1^i \frac{N - | \iota |}{|\iota|^2} e^{- \alpha s_{max}(\Vec{d}^{\min}) / 2} 
        + 
        \frac{c_1^i c_2^i}{| \iota |} e^{-c_3^i \alpha} 
        . 
    \end{split}
    \end{equation}
    
    As a result, in any cases, 
    \begin{equation}
    \label{ineq: norm of softmax - p}
    \begin{split}
        \| \text{Softmax} (\alpha \Vec{b^{-\alpha}}) -\Vec{p} \|
        \le 
        \frac{N \max_{i \in \iota} ( c_1^i ) }{|\iota|} e^{-\alpha s_{max} (\Vec{d}^{\min}) / 2} + \frac{\max_{i \in \iota} (c_1^i c_2^i)}{\sqrt{| \iota |}} e^{-\min_{i \in \iota}(c_3^i) \alpha} .
    \end{split}
    \end{equation}

    For the norm $\| \Vec{b^{-\alpha}} \|$, since
    \begin{equation}
        0 \le (\Vec{b^{-\alpha}})_k 
        = B_{-\alpha} (\Vec{d^k})
        \le \min(\Vec{d^k}) + M e^{-\alpha s_{min}(\Vec{d^k})} \| \Vec{d^k} \|, 
    \end{equation}
    write $\Vec{\epsilon} = (M e^{-\alpha s_{min}(\Vec{d^1})} \| \Vec{d^1} \| , \cdots , M e^{-\alpha s_{min}(\Vec{d^N})} \| \Vec{d^N} \| )$, 
    \begin{equation}
    \label{ineq: norm of b^-alpha}
    \begin{split}
        \| \Vec{b^{-\alpha}} \| 
        &\le \| \Vec{d}^{\min} + \Vec{\epsilon} \| 
        \le \| \Vec{d}^{\min} \| + \| \Vec{\epsilon} \|  \\
        &\le \| \Vec{d}^{\min} \| + M \sqrt{N} \max_{1 \le k \le N} ( \| \Vec{d^k} \| )
        e^{-\alpha \min_{1 \le k \le N}( s_{min}(\Vec{d^k}) )}
    \end{split}
    \end{equation}

    Finally the last term $| \Vec{p} \cdot (\Vec{b^{-\alpha}} - \Vec{d}^{\min}) |$, 
    for $i \in \iota$, 
    \begin{equation}
    \begin{split}
        (\Vec{b^{-\alpha}} - \Vec{d}^{\min})_i 
        = B_{-\alpha} (\Vec{d^i}) - \min(\Vec{d^i})
        \le M e^{-\alpha s_{min}(\Vec{d^i}) } \| \Vec{d^i} \|,
    \end{split}
    \end{equation}
    which implies
    \begin{equation}
    \begin{split}
    \label{ineq: p dot b^-alpha - d^min}
        | \Vec{p} \cdot (\Vec{b^{-\alpha}} - \Vec{d}^{\min}) | \le M e^{-\alpha \min_{i \in \iota} ( s_{min}(\Vec{d^i}) )} \max_{i \in \iota} ( \| \Vec{d^i} \| ). 
    \end{split}
    \end{equation}

    Combining inequalities (\ref{ineq: norm of softmax - p}), (\ref{ineq: norm of b^-alpha}), (\ref{ineq: p dot b^-alpha - d^min}) into (\ref{ineq: B_alpha (b^-alpha) - max(d^min)}), we have 
    \begin{equation}
    \begin{split}
        &\quad | B_\alpha (\Vec{b^{-\alpha}}) - \max(\Vec{d}^{\min}) | \\
        &\le
        \left( 
            \frac{N \max_{i \in \iota} ( c_1^i ) }{|\iota|} e^{-\alpha s_{max} (\Vec{d}^{\min}) / 2} + \frac{\max_{i \in \iota}(c_1^i c_2^i) }{\sqrt{| \iota |}} e^{-\min_{i \in \iota}(c_3^i) \alpha} 
        \right) \\
        &\quad \left(
            \| \Vec{d}^{\min} \| + M \sqrt{N} \max_{1 \le k \le N} ( \| \Vec{d^k} \| )
            e^{-\alpha \min_{1 \le k \le N}( s_{min}(\Vec{d^k}) )}
        \right) \\
        &\quad +
        M e^{-\alpha \min_{i \in \iota} ( s_{min}(\Vec{d^i}) )} \max_{i \in \iota} ( \| \Vec{d^i} \| ) \\
        &\le
        C_1 e^{-C_2 \alpha}
    \end{split}
    \end{equation}
    for some constants $C_1, C_2$ depending on $M, N, \{ d^i_j \}_{1 \le i \le N, 1 \le j \le M}$. 
    
    By a similar argument, $| B_\alpha (\Vec{b_{-\alpha}}) - \max(\Vec{d}_{\min}) | \le C_3 e^{-C_4 \alpha}$ can be shown. 
    Therefore, 
    \begin{equation}
        d'_H( f( \mathcal{M} ), \Omega ) - H_\alpha(\mathcal{Y}, \mathcal{W})
        = 
        O( e^{-\alpha C}), 
    \end{equation}
    and hence 
    \begin{equation}
        d_H( f( \mathcal{M} ), \Omega ) - H_\alpha(\mathcal{Y}, \mathcal{W})
        \le 
        h^{f(\mathcal{M})}_\mathcal{Y} + h^\Omega_\mathcal{W} + O( e^{-\alpha C}).
    \end{equation}
\end{proof}

The theorem suggests that $H_\alpha$ provides an upper bound for $d_H$. By minimizing $H_\alpha$, $d_H$ will decrease as well. Also, a large choice of $\alpha$ can reduce the uncertainty of the estimation.  



Besides for parameter domain matching, the proposed loss function HAND can also be applied to landmark matching. 
Some previous works \cite{lam2014landmark, choi2015flash} constraint the landmark correspondence in their models. 
On the other hand, some landmark matching tasks are achieved by 
minimizing a landmark matching energy $L$ \cite{lyu2024bijective, lyu2024spherical, zhang2022unifying} in the form of:
\begin{equation}
    L \left( f, \{ p_k \}_{k=1}^M, \{ q_k \}_{k=1}^M \right)
    = 
    \sum_{k=1}^M d \left( f(p_k), q_k \right),
\end{equation}
where $d$ is a distance metric. This approach matches only point-wise correspondence for a equal number of points in landmark and target. 

Here, with the presence of Hausdorff distance, we can match between landmarks with uneven number of points in landmarks 
$\mathcal{P} = \{ p_i \}_{i=1}^p$ 
and target positions
$\mathcal{Q} = \{ q_j \}_{j=1}^q$
by minimizing
\begin{equation}
\begin{split}
    L \left( f, \mathcal{P}, \mathcal{Q} \right)
    = 
    d_H \left( f( \mathcal{P} ), \mathcal{Q} \right), 
\end{split}
\end{equation}
or by minimizing our proposed HAND as
\begin{equation}
\begin{split}
    L \left( f, \mathcal{P}, \mathcal{Q} \right)
    = 
    H_\alpha \left( f( \mathcal{P} ), \mathcal{Q} \right). 
\end{split}
\end{equation}
More generally, we can extend the approach for region-wise correspondence. Given $M$ pairs of landmark regions 
$\{ P_k \}_{k = 1}^M$ 
and their corresponding target regions 
$\{ Q_k \}_{k=1}^M$, 
\begin{equation}
\label{eqn: landmark hand}
    L \left( f, \{ \mathcal{P}_k \}_{k=1}^M, \{ \mathcal{Q}_k \}_{k=1}^M \right)
    = 
    \sum_{k=1}^M H_\alpha \left( f(\mathcal{P}_k), \mathcal{Q}_k \right), 
\end{equation}
where $\mathcal{P}_k \subseteq P_k$ and $\mathcal{Q}_k \subseteq Q_k$ are finite points sampled in the regions, possibly in different numbers. This approach for landmark matching generalized the traditional methods, allowing matching between regions or unbalanced numbers of points between landmarks and targets, provides much more flexibilities than before. 

\subsection{Localized Energy for Geometry}
The conformal mapping, in Definition \ref{def: conformal map}, of a infinitesimal neighbourhood of a point $x_i$ is scaled by its conformal factor. We propose an energy formulation, aiming to mimic this idea by optimizing a mapping $f$ and a scalar function $\lambda$ through minimizing a $\lambda$-scaled local distance distortion energy introduced below. 

Define a quantity measuring distortion of two points $x_i, x_j \in \mathcal{X}$ by:
\begin{equation}
    D_{ij} = \Bigl( e^{- \| x_i - x_j \|^2 / \sigma^2} - e^{- \| f(x_i) - f(x_j) \|^2 / \sigma^2 \lambda_{ij}^2 } \Bigr) ^2, 
\end{equation}
where $\sigma, \lambda_{ij} = \lambda(x_i, x_j) > 0$. 
Localized Energy for Geometry (LEG) for mapping $f$ on point cloud $\mathcal{X}$ defined as the average of all $D_{ij}$, 
\begin{equation}
\label{def: LDDE: full}
\begin{split}
    D_\sigma(f, \lambda, \mathcal{X}) 
    &= 
    \frac{1}{N^2} \sum_{x_i, x_j \in \mathcal{X}} \Bigl( e^{- \| x_i - x_j \|^2 / \sigma^2} - e^{- \| f(x_i) - f(x_j) \|^2 / \sigma^2 \lambda_{ij}^2 } \Bigr) ^2
\end{split}
\end{equation}

This geometry distortion energy puts a higher weight on the points that are close to each other. 
Specifically, note that $e^{- \| x_i - x_j \|^2 / \sigma^2}$ approaches $1$ as $\| x_i - x_j \| \to 0$, and approaches $0$ as $\| x_i - x_j \| \to \infty$. Therefore, for $x_i, x_j$ close to each other, $D_{ij}$ is small only if $\| f(x_i) - f(x_j) \|$ is sufficiently close to $\| x_i - x_j \|$ up to a factor of $\lambda_{ij}$, or $\sigma$ is chosen to be much smaller than $\| x_i - x_j \|$ and $\| f(x_i) - f(x_j) \|^2 / \lambda_{ij}^2$. On the other hand, if $x_i, x_j$ are far away from each other, $D_{ij}$ will be small as long as $f(x_i), f(x_j)$ are also distant from each other, and the difference between $\| x_i - x_j \|$ and $\| f(x_i) - f(x_j) \|^2 / \lambda_{ij}^2$ will play a less important role in the overall energy. 


Suppose a triangular mesh on the point cloud describes a surface structure, in the sense that each edge appears in at most two triangular faces, denoted as $(\mathcal{X}, \mathcal{T})$, where $\mathcal{T}$ is the set of the triangles. Each element $T \in \mathcal{T}$ is an index set $\{ i, j, k \}$ containing the indices of vertices. This structure need not to be the underlying ground truth triangular mesh, any arbitrary surface mesh satisfies the theorem below, provided a fidelity assumption and a smoothness assumption on $\lambda$. 
Let $\mathcal{Y} = f(\mathcal{X}) = \{ y_i \}_{i=1}^N$, write $\lambda^{true}_{ij} = \| y_i - y_j \| / \| x_i - x_j \|$.

\begin{figure}[!h]
    \centering
    \includegraphics[width=.45\textwidth]{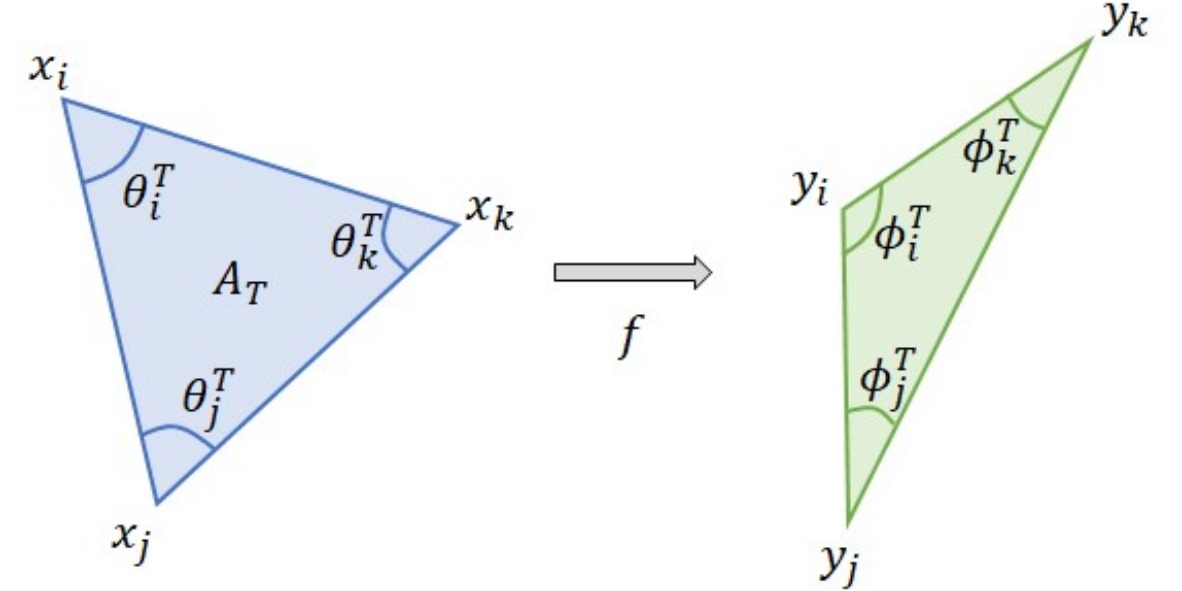}
    \caption{A figure illustrating the deformation of a triangle $T$ under mapping $f$.}
    \label{fig: LDDE illustration}
\end{figure}

\begin{theorem}
\label{thm: LDDE > angle distortion}
    Under the above notations, 
    write the longest edge length in $(\mathcal{X}, \mathcal{T})$ to be $R$. 
    Suppose there exist constants $0 < r_s, r_f < 1$ such that
    for any edge $i, j$,
    $|\lambda^{true}_{ij} / \lambda_{ij} - 1| < r_f$, and
    for any triangle $T = \{ i, j, k \}\in \mathcal{T}$, $|\lambda_{ik} / \lambda_{ij} - 1| < r_s$, 
    then
    \begin{equation}
    \label{ineq: LDDE > angle distortion}
    \begin{split}
        &\quad 
        \left(
            \sigma^2
            e^{
                (1 + r_f)^2 
                R^2
                / 
                \sigma^2
            }
        \right)^2
        D_\sigma
        + 
        \frac{162 |\mathcal{T}|} {N^2} 
        R^4 ( r_s + r_f )
        \\ &\ge
        \frac{4}{\pi^2}
        \frac{1}{N^2} 
        \sum_{\{ i, j, k \} = T \in \mathcal{T}}
        A_T^2
        \left[
            \left( 
                \theta^T_i - \phi^T_i 
            \right)^2
            + 
            \left( 
                \theta^T_j - \phi^T_j 
            \right)^2
            + 
            \left( 
                \theta^T_k - \phi^T_k 
            \right)^2
        \right] 
    \end{split}
    \end{equation}
    where $A_T$ is the area of triangle formed by $x_i, x_j, x_k$, $\theta^T_i = \angle x_k x_i x_j, \phi^T_i = \angle y_k y_i y_j$. 
\end{theorem}

\begin{proof}
    Write $d^x_{ij} = \| x_i - x_j \|, d^y_{ij} = \| y_i - y_j \|$,
    $\eta = 1/\sigma$. 
    By mean value theorem, 
    \begin{equation}
    \begin{split}
        D_{ij} &= \left( e^{-(d^x_{ij})^2 / \sigma^2} - e^{-(d^y_{ij})^2 / \sigma^2 \lambda_{ij}^2} \right)^2 
        \\ &\ge
        \left( 
            (d^x_{ij})^2 - (d^y_{ij})^2 / \lambda_{ij}^2
        \right)^2
        \left(
            \eta^2 
                e^{
                    -\eta^2 (1 + r_f)^2 
                    (d^x_{ij})^2
                }
        \right)^2
        .
    \end{split}
    \end{equation}
    
    Pick $T = \{ i, j, k \} \in \mathcal{T}$. 
    For simpler notations, write the edge representations $e^T_1 = [ij]$, $e^T_2 = [jk]$, $e^T_3 = [ki]$, with a period of $3$, i.e. $e^T_n = e^T_{n + 3}$. 
    Since 
    $a^2 + b^2 + c^2 = 
    \frac{1}{4} (a + b + c)^2 + 
    \frac{1}{4} (a + b - c)^2 + 
    \frac{1}{4} (a - b + c)^2 + 
    \frac{1}{4} (a - b - c)^2$
    , 
    \begin{equation}
    \label{ineq: Dij + Djk + Dki >= three terms}
    \begin{split}
        &\quad D_{ij} + D_{jk} + D_{ki}
        = D_{e^T_1} + D_{e^T_2} + D_{e^T_3}
        \\ &\ge
        \left( 
            \eta^2 
            e^{
                -\eta^2 (1 + r_f)^2 
                R^2
            }
        \right)^2
        \sum^{3}_{n=1}
            \left( 
                (d^x_{e^T_n})^2 - (d^y_{e^T_n})^2 / \lambda_{e^T_n}^2
            \right)^2
        \\
        &= 
        \frac{1}{4}
        \mathcal{E}^{-2}
        \sum^{3}_{n=1}
        (\mathcal{D}^T_n)^2
    \end{split}
    \end{equation}
    where $\mathcal{E} = e^{\eta^2 (1 + r_f)^2 R^2} / \eta^2$, and 
    \begin{equation}
    \begin{split}
        \mathcal{D}^T_n 
        =
        \sum^{2}_{m=0}
            (-1)^m
            \left(
            \left(
            d^x_{ e^T_{n+m} }
            \right)
            ^2 - 
            \left(
                d^y_{ e^T_{n+m} }
                / 
                \lambda_{ e^T_{n+m} }
            \right) ^ 2
            \right). 
    \end{split}
    \end{equation}
    
    Below, consider only the $\mathcal{D}_2^T$ in the inequality (\ref{ineq: Dij + Djk + Dki >= three terms}), and the others follow directly by symmetry.

    \begin{equation}
    \label{eqn: cosine distortion + e1 + e2 + e3}
    \begin{split}
    \mathcal{D}_2^T
    &=
    \left( 
        (d^x_{ij})^2 - (d^y_{ij})^2 / \lambda_{ij}^2
    \right)
    + 
    \left( 
        (d^x_{jk})^2 - (d^y_{jk})^2 / \lambda_{jk}^2
    \right)
    - 
    \left( 
        (d^x_{ki})^2 - (d^y_{ki})^2 / \lambda_{ki}^2
    \right)
    \\ &=
    2 d^x_{ij} d^x_{jk} (\cos \theta_j - \cos \phi_j)
    +
    [
        (d^y_{ij})^2 + (d^y_{jk})^2 - 2 d^y_{ij} d^y_{jk} \cos \phi_j 
    ] / \lambda_{ki}^2
    \\ &\quad
    - 
    [
        (d^y_{ij})^2 / \lambda_{ij}^2 
        + 
        (d^y_{jk})^2 / \lambda_{jk}^2 
        - 
        2 d^x_{ij} d^x_{jk} \cos \phi_j 
    ]
    \\ &=
    [
        2 d^x_{ij} d^x_{jk} (\cos \theta_j - \cos \phi_j)
    ]
    + 
    [
        (d^y_{ij})^2 / \lambda_{ki}^2
        - 
        (d^y_{ij})^2 / \lambda_{ij}^2
    ]
    \\ &\quad
    +
    [
        (d^y_{jk})^2 / \lambda_{ki}^2 
        - 
        (d^y_{jk})^2 / \lambda_{jk}^2
    ]
    +
    [
        2 d^x_{ij} d^x_{jk} \cos \phi_j 
        - 
        2 d^y_{ij} d^y_{jk} \cos \phi_j / \lambda_{ki}^2
    ]
    \\ &= 
    \mathcal{I}_0 + \mathcal{I}_1 + \mathcal{I}_2 + \mathcal{I}_3
    \end{split}
    \end{equation}
    where $\mathcal{I}_0, \mathcal{I}_1, \mathcal{I}_2, \mathcal{I}_3$ are the four terms in squared brackets respectively. 


    First, note that
    \begin{equation}
    \begin{split}
        |\mathcal{I}_0| = 
        \left|
        2 d^x_{ij} d^x_{jk} 
        \left( 
            \cos \theta_j - \cos \phi_j 
        \right)
        \right| 
        \le
        4 R^2. 
    \end{split}
    \end{equation}

    The terms $\mathcal{I}_1$ and $\mathcal{I}_2$ are symmetric and only the former is analyzed. Since $1 - ab = (1 - a) + (a - 1)(1 - b) + (1 - b)$, 
    \begin{equation}
    \begin{split}
        \mathcal{I}_1 
        &= 
        (d^x_{ij})^2 
        \left(
        \frac{\lambda^{true}_{ij}}{\lambda_{ij}}
        \right)^2
        \left(
        \frac{\lambda_{ij}^2}{\lambda_{ik}^2} - 1
        \right)
        \\ &=
        (d^x_{ij})^2 
        \left(
        \frac{(\lambda^{true}_{ij})^2}
        {\lambda_{ij}^2}
        - 1
        \right)
        \left(
        \frac{\lambda_{ij}^2}
        {\lambda_{ik}^2} 
        - 1
        \right)
        + 
        (d^x_{ij})^2 
        \left(
        \frac{\lambda_{ij}^2}
        {\lambda_{ik}^2} 
        - 1
        \right). 
    \end{split}
    \end{equation}
    Then
    \begin{equation}
    \begin{split}
        |\mathcal{I}_1| 
        &\le 
        R^2 (2 r_f + r_f^2) (2 r_s + r_s^2) + R^2 (2 r_s + r_s^2)
        \le 
        6 R^2 (r_f + r_s)
        .
    \end{split}
    \end{equation}
    The same estimation bound holds for $|\mathcal{I}_2|$. 

    For the last term $\mathcal{I}_3$, 
    \begin{equation}
    \begin{split}
        \mathcal{I}_3 
        &= 
        2 d^x_{ij} d^x_{jk} \cos \phi_j 
        \left(
            1 - 
            \frac{\lambda^{true}_{ij}}{\lambda_{ij}}
            \frac{\lambda^{true}_{jk}}{\lambda_{jk}}
            \frac{\lambda_{ij}}{\lambda_{ki}}
            \frac{\lambda_{jk}}{\lambda_{ki}}
        \right)
        \\ &= 
        2 d^x_{ij} d^x_{jk} \cos \phi_j 
        \Bigg[
        \left(
            1 - 
            \frac{\lambda^{true}_{ij}}{\lambda_{ij}}
            \frac{\lambda^{true}_{jk}}{\lambda_{jk}}
        \right)
        +
        \left(
            \frac{\lambda^{true}_{ij}}{\lambda_{ij}}
            \frac{\lambda^{true}_{jk}}{\lambda_{jk}}
            - 1
        \right)
        \left(
            1 - 
            \frac{\lambda_{ij}}{\lambda_{ki}}
            \frac{\lambda_{jk}}{\lambda_{ki}}
        \right)
        \\ &\quad
        + 
        \left(
            1 - 
            \frac{\lambda_{ij}}{\lambda_{ki}}
            \frac{\lambda_{jk}}{\lambda_{ki}}
        \right)
        \Bigg],
    \end{split}
    \end{equation}
    and 
    \begin{equation}
    \begin{split}
        |\mathcal{I}_3| 
        &\le 
        2 R^2 
        [
            (2 r_f + r_f^2)
            +
            (2 r_f + r_f^2)
            (2 r_s + r_s^2)
            +
            (2 r_s + r_s^2)
        ]
        \le
        15 R^2
        (r_f + r_s)
        .
    \end{split}
    \end{equation}

    By some computations using mean value theorem and compound angle formula, we have
    \begin{equation}
    \begin{split}
        (\cos \theta - \cos \phi)^2 \ge \frac{1}{\pi^2} \sin^2 \theta (\theta - \phi)^2. 
    \end{split}
    \end{equation}

    

    Also, note $(a + b)^2 + \mathcal{C} \ge a^2 + b^2$ if $\mathcal{C} \ge 2|ab|$. 
    By taking 
    \begin{equation}
    \begin{split}
        \mathcal{C} 
        &= 216 R^4 (r_f + r_s),
    \end{split}
    \end{equation}
    we have 
    \begin{equation}
    \begin{split}
        (\mathcal{D}_2^T)^2 + \mathcal{C}
        &=
        \left[
            \mathcal{I}_0 
            + 
            ( \mathcal{I}_1 + \mathcal{I}_2 + \mathcal{I}_3 )
        \right]^2
        + C
        \\ &\ge
        \left(
            2 d^x_{ij} d^x_{jk} 
        \right)^2
        \left( 
            \cos \theta_j - \cos \phi_j 
        \right)^2
        + 
        ( 
            \mathcal{I}_1 + \mathcal{I}_2 + \mathcal{I}_3 
        )^2
        \\ &\ge
        \frac{16}{\pi^2}
        A_{T}^2
        \left( 
            \theta_j - \phi_j 
        \right)^2
        + 
        ( 
            \mathcal{I}_1 + \mathcal{I}_2 + \mathcal{I}_3 
        )^2
    \end{split}
    \end{equation}
    where $A_{T}$ is the area of triangle formed by $x_i, x_j, x_k$. By symmetry, $\mathcal{D}_1^T$ and $\mathcal{D}_3^T$ are estimating the distortions of angle at $i$ and $k$ with respectively. 
    
    So the inequality in (\ref{ineq: Dij + Djk + Dki >= three terms}) can be rewritten as:
    \begin{equation}
    \begin{split}
        \mathcal{E} ^ 2
        \sum_{n = 1}^3 D_{e^T_n}
        + \frac{3}{4} \mathcal{C}
        \ge 
        \frac{4}{\pi^2}
        A_{T}^2
        \left[
        \left( 
            \theta_i - \phi_i 
        \right)^2
        + 
        \left( 
            \theta_j - \phi_j 
        \right)^2
        + 
        \left( 
            \theta_k - \phi_k 
        \right)^2
        \right]. 
    \end{split}
    \end{equation}
    
    Note that 
    \begin{equation}
    \begin{split}
        \mathcal{E} ^ 2
        D_\sigma
        + 
        \frac{
            3 | \mathcal{T} |
        }{
            4 N^2
        } 
        \mathcal{C}
        &= 
        \frac{ \mathcal{E}^2 }{N^2} \sum_{i, j = 1}^{N} D_{ij}
        + 
        \frac{
            3 | \mathcal{T} |
        }{
            4 N^2
        } 
        \mathcal{C}
        \ge
        \frac{1}{N^2} \sum_{T \in \mathcal{T}} 
        \left[
            \mathcal{E}^2 
            \sum_{n = 1}^{3} 
            D_{e^T_n}
            + 
            \frac{3}{4} \mathcal{C}
        \right]
        . 
    \end{split}
    \end{equation}
    Altogether, 
    \begin{equation}
    \label{ineq: sth*D + O(r^4 R^4) > angle distortion}
    \begin{split}
        &\quad 
        \left(
            \sigma^2
            e^{
                (1 + r_f)^2 
                R^2 / \sigma^2
            }
        \right)^2
        D_\sigma
        + 
        \frac{162 |\mathcal{T}|} {N^2} 
        R^4 ( r_f + r_s )
        \\ &
        \ge
        \frac{4}{\pi^2}
        \frac{1}{N^2} 
        \sum_{\{ i, j, k \} = T \in \mathcal{T}}
        A_T^2
        \left[
            \left( 
                \theta^T_i - \phi^T_i 
            \right)^2
            + 
            \left( 
                \theta^T_j - \phi^T_j 
            \right)^2
            + 
            \left( 
                \theta^T_k - \phi^T_k 
            \right)^2
        \right] 
        .
    \end{split}
    \end{equation}
\end{proof}

The coefficient $\sigma^2 e^{ (1 + r_f)^2 R^2 / \sigma^2}$ attains minimum when $\sigma = ( 1 + r_f ) R$ with minimum value $(1 + r_f)^2 R^2 e$. Together with the second terms $ 162|\mathcal{T}| R^4 ( r_s + r_f ) / N^2 $, the theorem stated above suggests that a proper choice of $\sigma$ and small values in $r_f$ and $r_s$ provide a finer upper bound for the angle distortions between the original point cloud and its mapping. 
The relationship between the assumptions and angle distortion can be understood intuitively. 
By smoothness assumption, $\lambda_{ij} \approx \lambda_{jk} \approx \lambda_{ik}$, and by fidelity assumption, $\lambda_{ij} \approx \lambda_{ij}^{true}$, 
$\lambda_{jk} \approx \lambda_{jk}^{true}$, 
$\lambda_{ik} \approx \lambda_{ik}^{true}$, we have $\lambda_{ij}^{true} \approx \lambda_{jk}^{true} \approx \lambda_{ik}^{true}$, which means the two triangles are approximately similar, and hence low in angle distortion. Also, the difference in angle is accumulated from fidelity and smoothness error. 

Assuming the optimal choice $\sigma = (1 + r_f) R$, and denoting $\Theta_T = (\theta^T_i - \phi^T_i)^2 + (\theta^T_j - \phi^T_j)^2 + (\theta^T_k - \phi^T_k)^2$, dividing both sides by $R^4$, the inequality (\ref{ineq: LDDE > angle distortion}) can be rewritten as:

\begin{equation}
\begin{split}
    (1 + r_f)^4 e^2
    D_\sigma
    + 
    \frac{162 |\mathcal{T}|} {N^2} 
    ( r_s + r_f )
    &\ge
    \frac{4}{N^2} 
    \sum_{T \in \mathcal{T}}
    \left(
        \frac{A_T}{\pi R^2}
    \right)^2
    \Theta_T
    \\ &=
    \frac{4}{N^2} 
    \sum_{T \in \mathcal{T}}
    \left(
        \frac{A_T}{\pi R_T^2}
    \right)^2
    \left(
        \frac{R_T}{R}
    \right)^4
    \Theta_T,
\end{split}
\end{equation}
where $R_T$ denotes the longest edge length of triangle $T$. The factor $A_T / \pi R_T^2$ serves as a triangle-to-circle area ratio on triangle $T$. This value is large if the angles of $T$ are uniform, i.e. close to an equilateral triangle, and it is small if the angles are irregular. Meanwhile, $R_T / R$ is a local-to-global edge length ratio between the longest edge length on $T$ to the longest edge length on the entire mesh $(\mathcal{X}, \mathcal{T})$. So Theorem \ref{thm: LDDE > angle distortion} puts a higher weight on angle distortion in triangle with higher triangle-to-circle area ratio and higher local-to-global edge length ratio. 
Intuitively, although Theorem \ref{thm: LDDE > angle distortion} holds for all possible meshes defined on $\mathcal{X}$, the minimization of angle distortions on meshes with more uniform angle and length distribution is preferred, and large angle distortion on nonsensical meshes is allowed. In Figure \ref{fig: leg factors}, the two left plots show how variance in angle affects $A_T / \pi R_T^2$, given the same $R_T$, the two right plots show how different mesh structures defined on the same point cloud affect $R$ and local-to-global edge length ratio. 

\begin{figure}[!h]
    \centering
    \includegraphics[width=.9\textwidth]{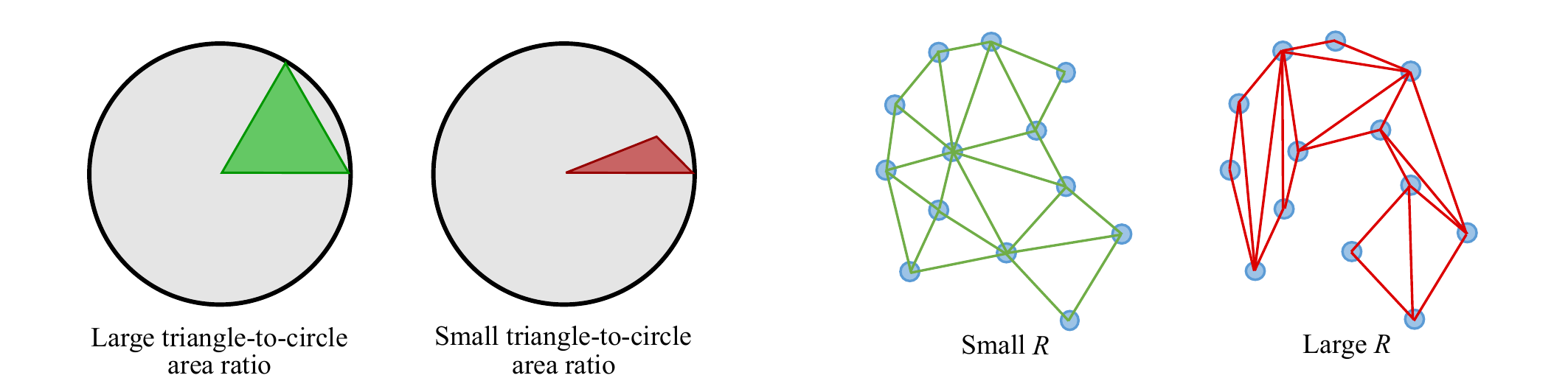}
    \caption{An illustration explaining how the mesh structure affects the factors.}
    \label{fig: leg factors}
\end{figure}



Moreover, 
if $\lambda$ has $N (N-1)$ degrees of freedom, then $D_\sigma = 0$ when $\lambda_{ij} = \| f(x_i) - f(x_j) \| / \| x_i - x_j \|$ for any arbitrary $f, i \ne j$, which is an undesirable local minimum. 
This phenomenon can be explained by Theorem \ref{thm: LDDE > angle distortion}, that the fidelity quantity $r_f$ is $0$, while the smoothness quantity $r_s$ can be very large, resulting in large angle distortion. 
For reducing the degree of freedom, and also for improving its smoothness, a particular formulation for $\lambda$ has to be chosen, which will be discussed below, in the Section \ref{sec: optim algo}.

\subsection{Optimization Algorithm}
\label{sec: optim algo}
With the discussed loss functions, we model the surface parametrization problem as a optimization problem, with three terms: LEG, HAND for parameter domain matching, and HAND for landmark matching, 
\begin{equation}
\label{optim problem}
    ( f^*, \lambda^* )
    = \argmin_{f, \lambda} 
    \beta_1 D_\sigma (f, \lambda, \mathcal{X}) + 
    \beta_2 H_\alpha ( f(\mathcal{X}), \mathcal{W} ) + 
    \beta_3 \sum_{k=1}^M H_\alpha \left( f(\mathcal{P}_k), \mathcal{Q}_k \right)
\end{equation}
where $\beta_1, \beta_2, \beta_3 > 0$, 
$\mathcal{X} \subseteq \mathcal{M}$ is the point cloud surface, $\mathcal{W} \subseteq \Omega$ is a set of finite points in the parameter domain, 
$\{ \mathcal{P}_k \}_{k = 1}^M$ and $\{ \mathcal{Q}_k \}_{k = 1}^M$ are landmark and target finite point sets. 
Also, in order to better regularize the geometry of the neighbourhood of the landmarks, the landmarks are included in the calculation of LEG. 
The symmetric bivariate function $\lambda : \mathcal{X} \times \mathcal{X} \to \mathbb{R}$ is modelled by a univariate function $\lambda_{inv}$ by $1 / \lambda (x_i, x_j) = \lambda_{inv} (x_i) + \lambda_{inv} (x_j)$, which is named inverse $\lambda$ function. 
The optimization problem \ref{optim problem} is finalized as 
\begin{equation}
\label{optim problem overall}
\begin{split}
    &\quad ( f^*, \lambda_{inv}^* )
    = \argmin_{f, \lambda_{inv}} \mathcal{L} (f, \lambda_{inv}, \mathcal{X}, \mathcal{W}, \{ \mathcal{P}_k \}_{k=1}^M, \{ \mathcal{Q}_k \}_{k=1}^M) 
    \\ &=
    \argmin_{f, \lambda_{inv}} \beta_1 D_\sigma \left( f, \lambda_{inv}, \mathcal{X}
    \cup \mathcal{P}
    \right) + \beta_2 H_\alpha ( f(\mathcal{X}), \mathcal{W} ) + \beta_3 \sum_{k=1}^M H_\alpha \left( f(\mathcal{P}_k), \mathcal{Q}_k \right)
    ,
\end{split}
\end{equation}
where $\mathcal{P} = \bigcup_{k=1}^M \mathcal{P}_k$. 

For the sake of continuity, we represent the mapping function and inverse $\lambda$ function using neural networks, giving the names parametrization net and inverse $\lambda$ net. Under the setting of $1 / \lambda(x_i, x_j) = \lambda_{inv}(x_i) + \lambda_{inv}(x_j)$, no matter the number of parameters in $\lambda_{inv}$-net, the degree of freedom of $\lambda_{ij}$ is at most $N$, resolving the potential problem stated in the last paragraph in previous section caused by the exceeding degree of freedom. 

Stochastic gradient descent (SGD) is also applied in the optimization. By the traditional idea of SGD, in $i$-epoch and $j$-batch, we sample subsets $\mathcal{X}^{i, j} \subseteq \mathcal{X}$ and $\mathcal{W}^{i, j} \subseteq \mathcal{W}$, and compute gradient of  
\begin{equation}
\label{loss in batch}
    \mathcal{L}^{i, j} = 
    \mathcal{L} (f, \lambda_{inv}, 
    \mathcal{X}^{i, j} 
    , \mathcal{W}^{i, j}, \{ \mathcal{P}_k \}_{k=1}^M, \{ \mathcal{Q}_k \}_{k=1}^M), 
\end{equation}
for updates in SGD. 
Here the landmarks and targets are not subsampled, as they have much smaller sizes compared to the point cloud $\mathcal{X}$ in general. 

\begin{figure}[!h]
    \centering
    \includegraphics[width=.7\textwidth]{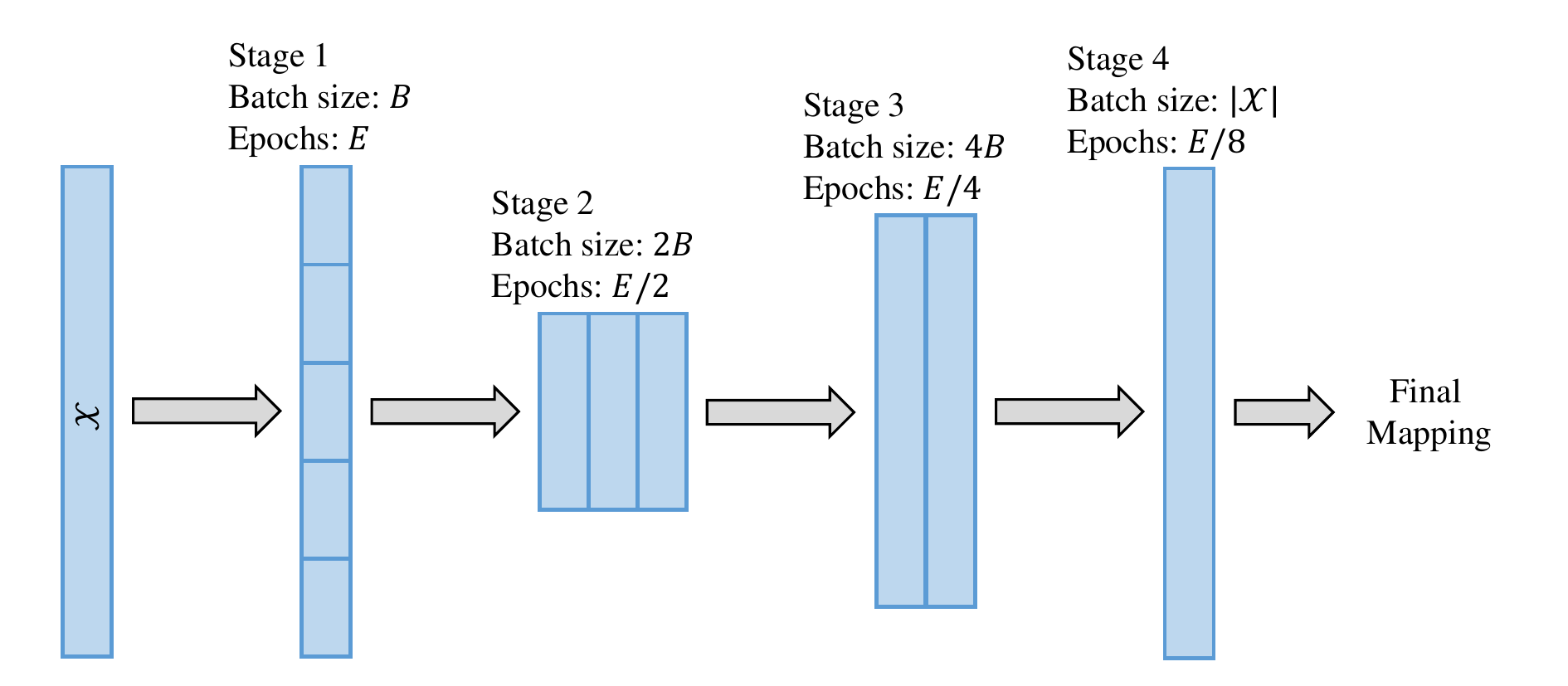}
    \caption{An illustration to the dynamics of the suggested optimization algorithm.}
    \label{fig: optim algo}
\end{figure}

Apart from the traditional meaning of approximation to the full gradient, Theorems \ref{thm: error bound for approximation of Hausdorff Distance} and \ref{thm: LDDE > angle distortion} proved previously provide a geometrical view for applying SGD. 
The point set $\mathcal{X}^{i, j} \cup \mathcal{P}$ can also be understood as a point cloud of the underlying surface $\mathcal{M}$, coarser than $\mathcal{X}$. The term $D_\sigma (f, \lambda_{inv}, \mathcal{X}^{i, j} \cup \mathcal{P})$ is not only an approximation to $D_\sigma (f, \lambda_{inv}, \mathcal{X} \cup \mathcal{P})$, but also a measurement of geometric distortion of the coarsened point cloud under mapping $f$. Theorem \ref{thm: LDDE > angle distortion} suggests that minimization on $D_\sigma (f, \lambda_{inv}, \mathcal{X}^{i, j} \cup \mathcal{P})$ reduces the angle distortion of an underlying surface mesh on $\mathcal{X}^{i, j} \cup \mathcal{P}$. 
On the other hand, by Theorem \ref{thm: error bound for approximation of Hausdorff Distance}, $H_\alpha ( f(\mathcal{X}^{i, j}), \mathcal{W}^{i, j})$ measures the alignment between the mapping $f(\mathcal{X}^{i, j})$ and point set $\mathcal{W}^{i, j}$. 
By increasing the size of $\mathcal{X}^{i, j}$ and $\mathcal{W}^{i, j}$, values of $h^{f(\mathcal{M})}_\mathcal{Y}$ and $h^{\Omega}_\mathcal{W}$ can be reduced. As a result, Theorem \ref{thm: error bound for approximation of Hausdorff Distance} suggests that the performance of domain matching can be improved. 

By design, the whole optimization process is divided into different stages, where the batch size is doubled after each stage, until the whole point cloud $\mathcal{X}$ is used. 
For a single stage, denote the number of epochs to be $E$, batch size in point cloud to be $B_\mathcal{X}$, batch size in parameter domain to be $B_\mathcal{W}$. 
The number of mini-batches in an epoch is set to be $B = \ceil{|\mathcal{X}| / B_\mathcal{X}}$. The computation of loss functions requires around $B_\mathcal{X}^2 + B_\mathcal{X} \cdot B_\mathcal{W}$ pairwise distances. 
Therefore the overall computation in a stage requires $E \cdot B \cdot (B_\mathcal{X}^2 + B_\mathcal{X} \cdot B_\mathcal{W}) \approx E \cdot |\mathcal{X}| \cdot (B_\mathcal{X} + B_\mathcal{W})$ pairs of distances. To compensate the rise in computations caused by doubling the batch size, the number of epochs is halved after each stage. This setting can be intuitively understood as fitting general pattern using coarser point clouds at first, then fine-tuning the details using the whole point cloud. The idea of this mechanism is illustrated in Figure \ref{fig: optim algo}. 

\begin{algorithm}[htb]
    \caption{Optimization Scheme}
    \label{algo: optim}
    \begin{algorithmic}[1]
        \Require
            Data: a point cloud $\mathcal{X}$, point set from parameter domain $\mathcal{W}$, landmark points $\{ \mathcal{P}_k \}_{k=1}^M$ and targets $\{ \mathcal{Q}_k \}_{k=1}^M$ ;
            Loss function weightings: $\beta_1$, $\beta_2$, $\beta_3$ ;
            Initial training hyperparameters: number of epochs $E$, batch sizes $B_\mathcal{X}$ and $B_\mathcal{W}$; LEG parameter $\sigma$ ;  
            HAND parameters, $\alpha_{initial}$, $\alpha_{final}$ ;
            Parameter bounds: $\sigma_{min}$, $\alpha_{max}$, $E_{min}$. 
        \Ensure
            Trained networks $f$ and $\lambda_{inv}$.
        \State 
            Initialize networks $f$, $\lambda_{inv}$;
        \State 
            $continue$ $\leftarrow$ $True$;
        \While{ $continue$ }
            \If{$B_{\mathcal{X}} = |\mathcal{X}|$} 
                \State
                    $continue$ $\leftarrow$ $False$ ; 
            \EndIf

            \State 
                $\alpha \leftarrow \alpha_{inital}$ ;
            \For{$i = 1, \cdots, E$}
                \For{$j = 1, \cdots, \lceil | \mathcal{X} | / B_{\mathcal{X}} \rceil$}
                    \State
                        Sample 
                        $\mathcal{X}^{i,j} \subseteq \mathcal{X}$, $\mathcal{W}^{i,j} \subseteq \mathcal{W}$, 
                        where
                        $ |\mathcal{X}^{i,j}| = B_\mathcal{X}$, 
                        $ |\mathcal{W}^{i,j}| = B_\mathcal{W}$ ;
                    \State
                        Compute $\mathcal{L}^{i, j}$ in 
                        (\ref{loss in batch}) and its gradient to update $f, \lambda_{inv}$ ;
                \EndFor
                \If{$i < E$} 
                    \State
                        $\alpha \leftarrow \alpha + (\alpha_{final} - \alpha_{initial}) / (E - 1)$  ;
                \EndIf
            \EndFor
            \State
                $\sigma \leftarrow \max ( \sigma / \sqrt{2}, \sigma_{min})$ ;
            \State
                $\alpha_{initial} \leftarrow \alpha_{final}$ ;
            \State
                $\alpha_{final} \leftarrow \min ( 2 \alpha_{final} , \alpha_{max})$ ;
            \State
                $E \leftarrow \max ( E / 2 , E_{min})$;
            \State
                $B_{\mathcal{X}} \leftarrow \min ( 2 B_{\mathcal{X}} , |\mathcal{X}|)$ ;
            \State
                $B_{\mathcal{W}} \leftarrow \min ( 2 B_{\mathcal{W}} , |\mathcal{W}|)$ ;
        \EndWhile
    \end{algorithmic}
\end{algorithm}

In Theorem \ref{thm: LDDE > angle distortion}, we observe that the optimal choice of $\sigma$ depends on $R$ and $r_f$. We expect the fidelity error $r_f$ decreases throughout the optimization, but the decreasing rate is hard to analyze. 
When the density of the point cloud increases, the gaps between the points shrink. If the batch size is doubled in each stage, then $R$, the longest edge length, is expected to be scaled down to $R/\sqrt{2}$. Therefore, the parameter $\sigma$ is also scaled to $\sigma/\sqrt{2}$ to match $R$. Moreover, Theorem \ref{thm: error bound for approximation of Hausdorff Distance} states the importance of a large choice of $\alpha$, and so the values of $\alpha$ is also doubled after each stage. In practice, the increment of $\alpha$ is done by a linear increase, from $\alpha_{initial}$ at the stage of a stage, to the desired values $\alpha_{final}$. 
Depending on the number of stages, parameters may become too extreme, leading to numerical instability. In our algorithm, capping the values of $\alpha, \sigma$, and also the number of epochs, are performed to avoid the potential problems.  

The overall optimization algorithm is summarized in Algorithm \ref{algo: optim}.

%% file: sections/sec6-implementation.tex
\section{Implementation}
\label{sec: implementation}

The proposed loss functions and optimization algorithm were implemented in Python, using deep learning library PyTorch \cite{paszke2019pytorch}. The experiments were computed using an NVIDIA RTX A6000. 

The parameters used in the algorithm are stated here. If the associated loss function is used, LEG weight $\beta_1 = 5$, HAND weight $\beta_2 = 1$ for matching parameter domain, landmark mismatch energy weight $\beta_3 = 1$. At the beginning of the first stage in the optimization process, number of epochs $E = 10000$, batch sizes $B_\mathcal{X} = B_\mathcal{W} = 1024$, LEG parameter $\sigma = 0.5$, HAND parameters $\alpha_{initial} = 2$, $\alpha_{final} = 20$. We bound the parameters by $\sigma_{min} = 0.001$, $\alpha_{max} = 100$, $E_{min} = 1000$. 

\begin{figure}[!h]
    \centering
    \includegraphics[width=.8\textwidth]{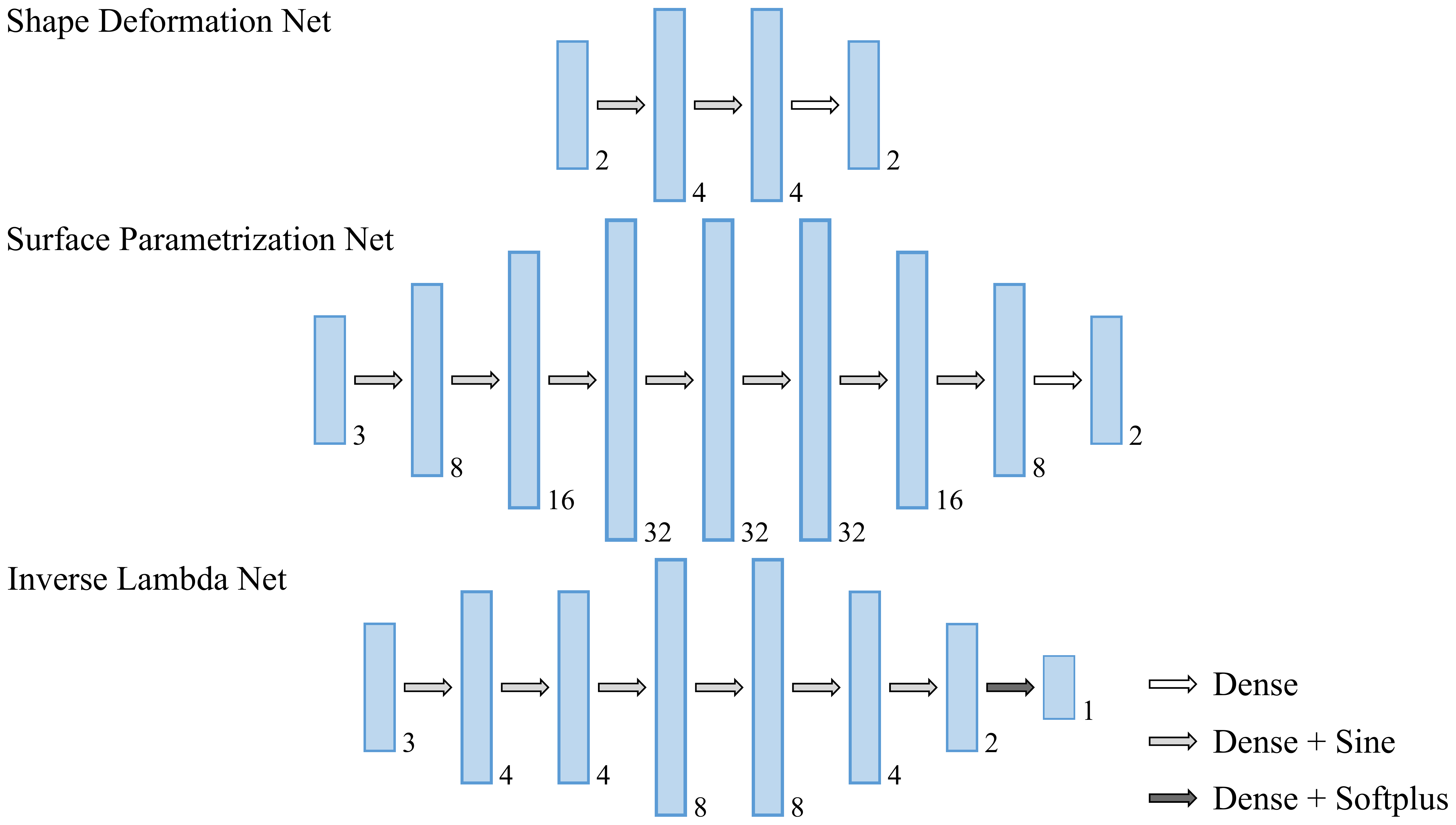}
    \caption{
        Architecture of the neural networks used in experiments. The arrows indicate the forward process. The bars indicate the latent variables, with dimensions stated under them. 
    }
    \label{fig: neural net arch}
\end{figure}



Networks of difference sizes are used for varying levels of difficulty, and their structure is illustrated in Figure \ref{fig: neural net arch}.  In the experiments in Section \ref{sec: plane shape deform}, a shape deforming network, shown in the top of Figure \ref{fig: neural net arch}, is used in optimizing only HAND and no LEG is involved. 
In Sections \ref{sec: FB surf param}, free-boundary parametrizations are performed, and in Sections \ref{sec: DC surf param}, \ref{sec: DLC surf param}, surface parametrization experiments optimizing both HAND and LEG are carried. 
Larger fully-connected neural networks to model the parametrization functions, the middle one in Figure \ref{fig: neural net arch}, and $\lambda_{inv}$ functions, the bottom one in Figure \ref{fig: neural net arch}, are used. In the networks, the sine function is used as the activation, except for the output layers. No activation is needed for the output layer of the parametrization nets, while the softplus function is used to activate the last layer of the $\lambda_{inv}$ nets so that the output is always positive. The minimization is done by stochastic gradient descent with RMSprop \cite{tieleman2012lecture} with parameters $\alpha = 0.99$, a learning rate of $0.0001$, with a momentum of $0.9$. 
All the weights and biases in the networks are randomly initialized, and hence the initial mappings and $\lambda_{inv}$ values.

%% file: sections/sec7-experiment.tex
\section{Experimental results}
\label{sec: experiments}

In this section, we illustrate the performance of our proposed methods through experiments conducted with various combinations of loss function terms, each serving different purposes. Additionally, different parameter domains, including various shapes and topologies, are used in experiments to demonstrate the flexibility of our methods.


\subsection{Performance of HAND}
\label{sec: plane shape deform}

\begin{figure}[!h]
    \centering
    \includegraphics[width=.8\textwidth]{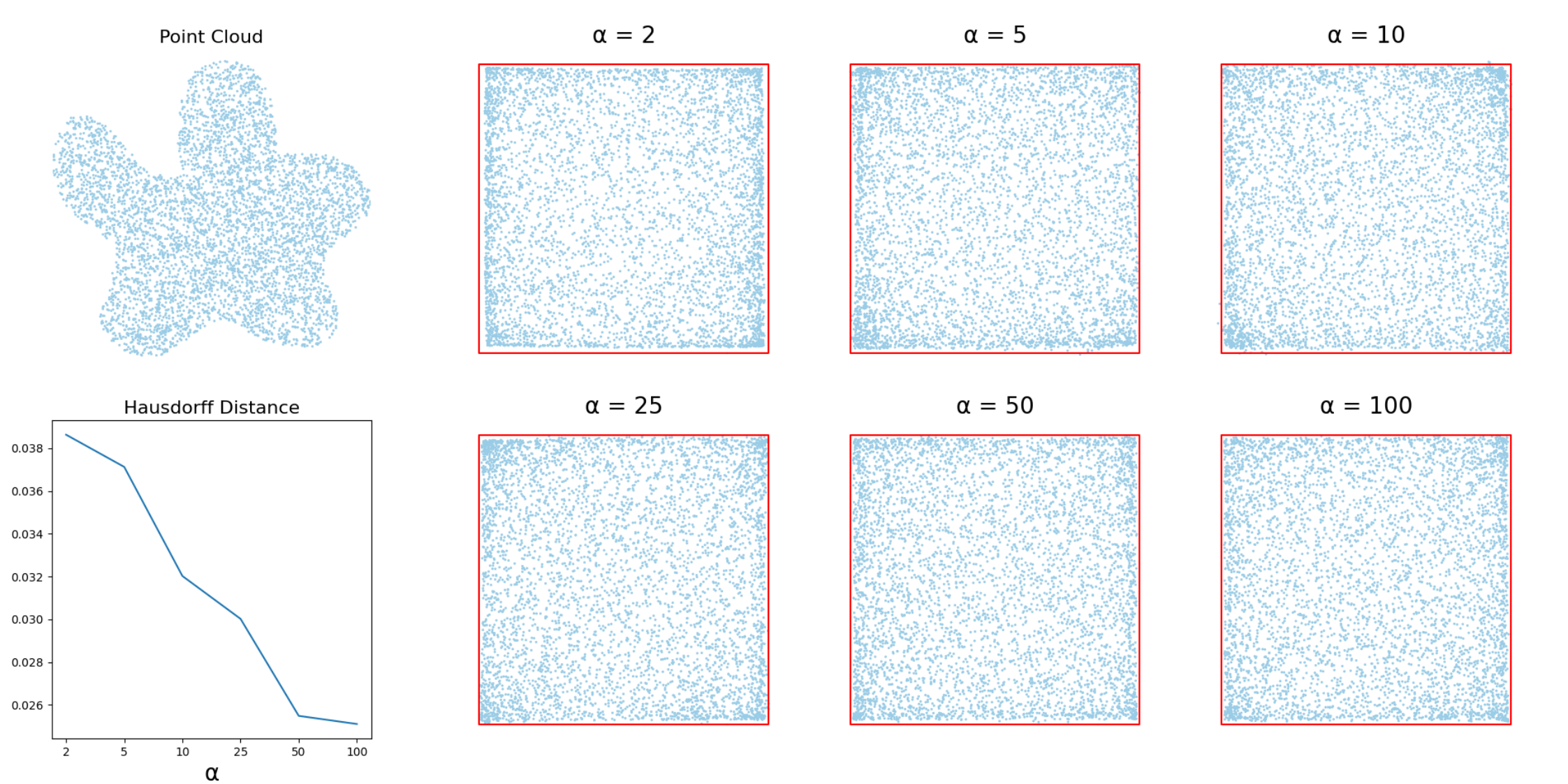}
    \caption{
        The mappings obtained by minimizing HAND with $\alpha = 2, 5, 10, 25, 50, 100$ on a 2D point cloud and a plot of their Hausdorff distance to a square, with boundary drawn in red. 
    }
    \label{fig: blob}
\end{figure}

We begin with an experiment for 2D shape deformation using only the proposed Hausdorff Approximation from Node-wise Distances (HAND). The experimental results are shown in Figure \ref{fig: blob}. 
A point cloud with an irregular shape, with $5000$ points in $\mathbb{R}^2$ is created, and mapped to a unit square, using only the presence of HAND, without considering the geometry distortion using LEG. In the figure, the boundary of unit square is drawn in red lines. 
We test the performance of HAND with different parameters $\alpha = 2, 5, 10, 25, 50, 100$. 
Due to the simplicity of the task, the optimizations are performed directly on entire point cloud, without dividing the optimization into several stages as discussed in Section \ref{sec: optim algo}.
The ability of shape matching by minimizing HAND is demonstrated in the Figure \ref{fig: blob}. Moreover, as value of $\alpha$ increases, a better match between the point cloud mapping and the boundary of the square is shown. 
To quantitatively examine the performance, a Hausdorff distance is measured between the mapped point cloud and a densely and uniformly sampled finite point set in the square. 
A line plot illustrating a decreasing trend in the Hausdorff distance as $\alpha$ increases. We can conclude that a larger value of $\alpha$ improves the performance of HAND in terms of matching matching the parameter domain. 

\subsection{Free-Boundary Surface Parametrization}
\label{sec: FB surf param}

\begin{figure}[!h]
    \centering
    \includegraphics[width=\textwidth]{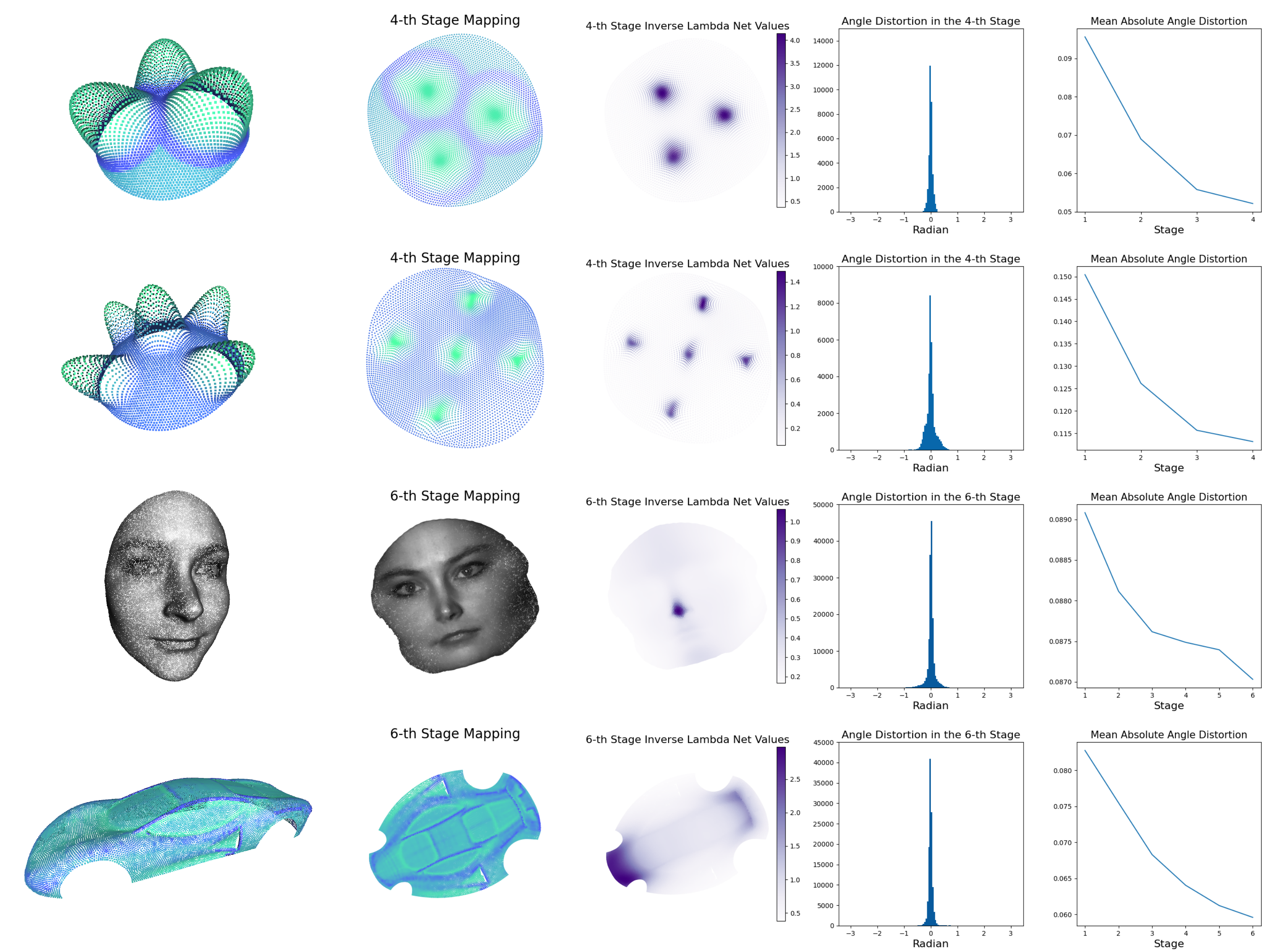}
    \caption{
        The results of free boundary experiments obtained by minimizing only LEG. Each row represent an experimental result for an input surface. 
        Top to bottom: three peaks, five peaks, human face, car shell. 
        Left to right: the input point cloud surface, final stage mapping, final stage inverse $\lambda$-Net values, histogram for angle distortion (in radian), and line plot of mean absolute angle distortion (in radian) against stages. 
    }
    \label{fig: freeBdy aio}
\end{figure}

Free-boundary parametrization for point cloud surfaces are done by minimization on only Localized Energy for Geometry (LEG).
We examine two open surfaces, one with three peaks and another with five peaks, both containing $5809$ points. Additionally, a human face mesh with $24891$ points, as well as a car shell surface with $19320$ points, are included in these free-boundary parametrization experiments. 

The parametrization results are illustrated in Figure \ref{fig: freeBdy aio}. Figure \ref{fig: freeBdy aio} shows the outputs of the parametrization networks $f(\mathcal{X})$ and inverse $\lambda$ networks $\lambda_{inv} (\mathcal{X})$ at the last stage. 
Theorem \ref{thm: LDDE > angle distortion} provides a relationship between LEG and angle distortion on triangular meshes. In order to quantitatively analyze the geometric distortion, a mesh structure is given to the point cloud, and angle distortion is computed on it. 
The histograms, which indicate that the angle differences in the last stage are concentrated at 0, suggest the ability of angle preserving via minimization of LEG. Also, the line plots illustrate decreasing trends in mean absolute angle distortion with the progression of stages, demonstrating that the optimization scheme yields better parametrization results.
The plots for $\lambda_{inv} (\mathcal{X})$ show that if the points are concentrated in certain regions, for example the peaks, nose, or the front part of the car shown in Figure \ref{fig: freeBdy aio}, the output values of $\lambda_{inv}$ nets are higher. This observation aligns with the intuition that $\lambda_{inv}$ measures the scaling of distance under the mapping $f$, and confirms the significance of $\lambda_{inv}$. 

\subsection{Domain-Constrained Surface Parametrization}
\label{sec: DC surf param}
Experiments incorporating both HAND and LEG in the optimization process are discussed in this section. The human face and car shell used previously in free-boundary parametrization experiments are also used in the following experiments. 

\begin{figure}[!h]
    \centering
    \includegraphics[width=.8\textwidth]{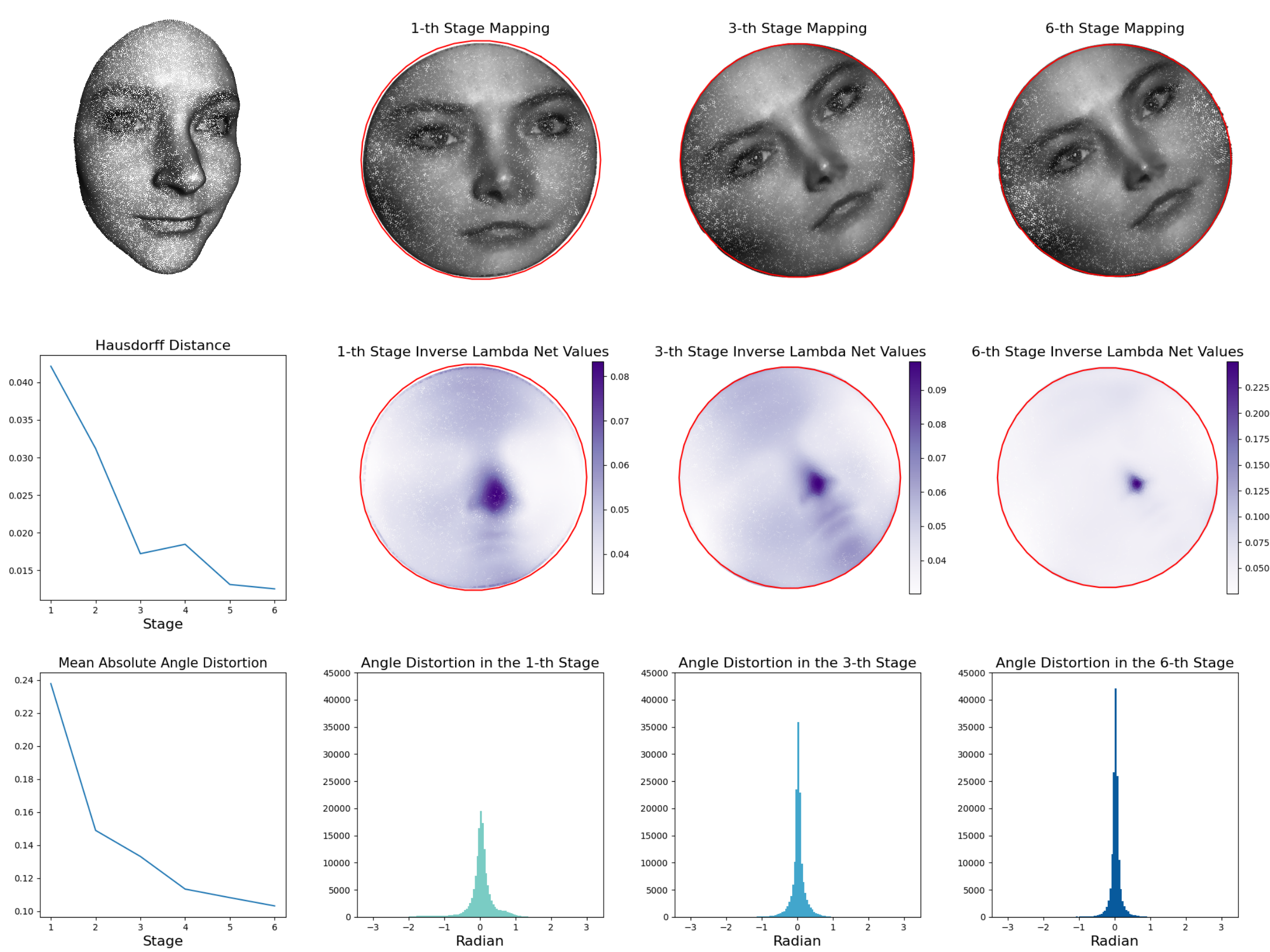}
    \caption{
        Results for parametrizing a human face to unit disk. In the first column, the input human face, line plot for Hausdorff distance against stages, line plot for mean absolute angle distortion (in radian) are plotted from top to bottom. The second to the last columns show the result after the $1^{\text{st}}$, $3^{\text{rd}}$ and $6^{\text{th}}$ stages, including output mappings, $\lambda_{inv}$ networks output values, histogram for angle distortion (in radian) from top to bottom. The red lines indicate the boundary of the unit disk.  
    }
    \label{fig: fixedBdy sophie}
\end{figure}

First of all, the parametrization on the human face is computed, with a unit disk chosen as the target parameter domain. 
The experimental results are shown in Figure \ref{fig: fixedBdy sophie}. This figure summarizes the overall performance of parametrization using the two line plots: one showing Hausdorff distance, the other showing mean absolute angle distortion (in radian), against stages. The plots generally show decreasing trends, suggesting the improvements made during the optimization process through parameter adjustments. 
The evolution from the earlier stages to latter stages are shown by plotting the results after the $1^{\text{st}}$, $3^{\text{rd}}$ and $6^{\text{th}}$ stages, with details in the last three columns. In the $1^{\text{st}}$ stage, a gap between the point cloud mapping and the boundary of the parameter domain can be clearly observed. Meanwhile, the alignment between the mapping and boundary is much better improved in latter stages, supporting the line plot for Hausdorff distance. The histograms for angle distortion (in radian) illustrate that the angle distortion are more concentrated at 0 in the latter stages than in the earlier stages, coinciding with the decrease in mean absolute angle distortion shown in the line plot. We can also observe that the mapping of the nose shrinks as the stages progress, matching the rising values in the nose illustrated in the plots for $\lambda_{inv}$ values. 

\begin{figure}[!h]
    \centering
    \includegraphics[width=.5\textwidth]{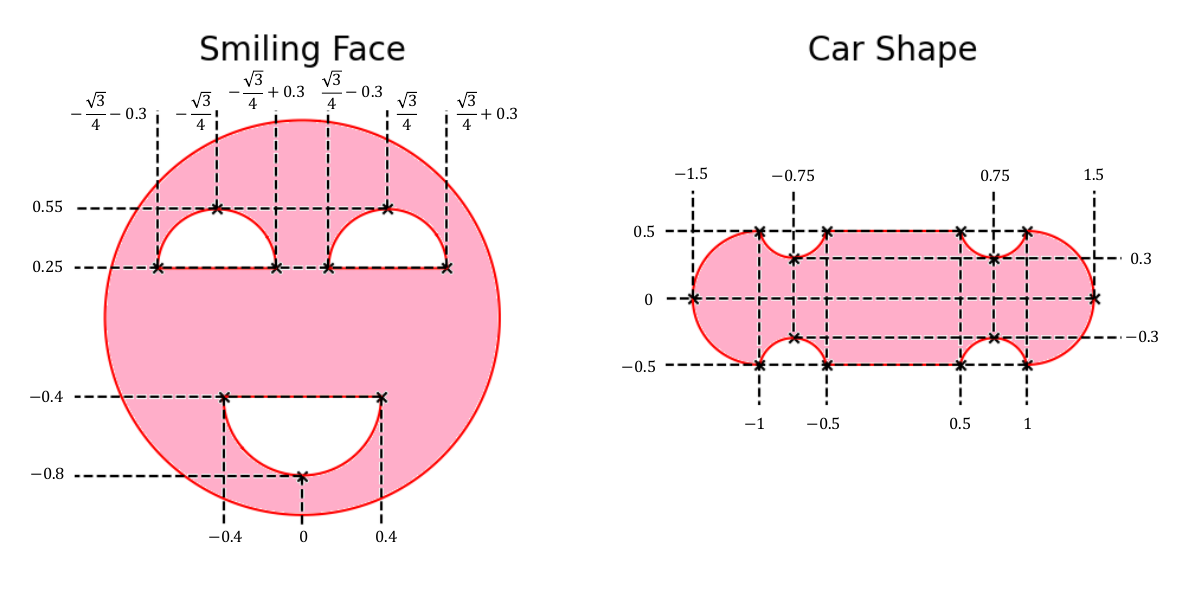}
    \caption{
        The smiling face and car-shaped parameter domains. The vertical and horizontal dash lines are indicating the $x$-coordinates and $y$-coordinates. 
    }
    \label{fig: param domain smiling car}
\end{figure}

Instead of choosing a simple parameter domain such as a unit disk, experiments with more complicated parameter domains are demonstrated. 
In Figure \ref{fig: param domain smiling car}, the left plot illustrated a multiply-connected smiling face, formed by removing half disks from unit disk, and the right plot shows a car-shaped parameter domain, whose boundary consist of circular arcs and straight lines. Figure \ref{fig: param domain smiling car} also indicates the coordinates for some important points in the shapes. The horizontal dash lines indicate the $y$-coordinates, while vertical dash lines indicate the $x$-coordinates. 

\begin{figure}[!h]
    \centering
    \includegraphics[width=.8\textwidth]{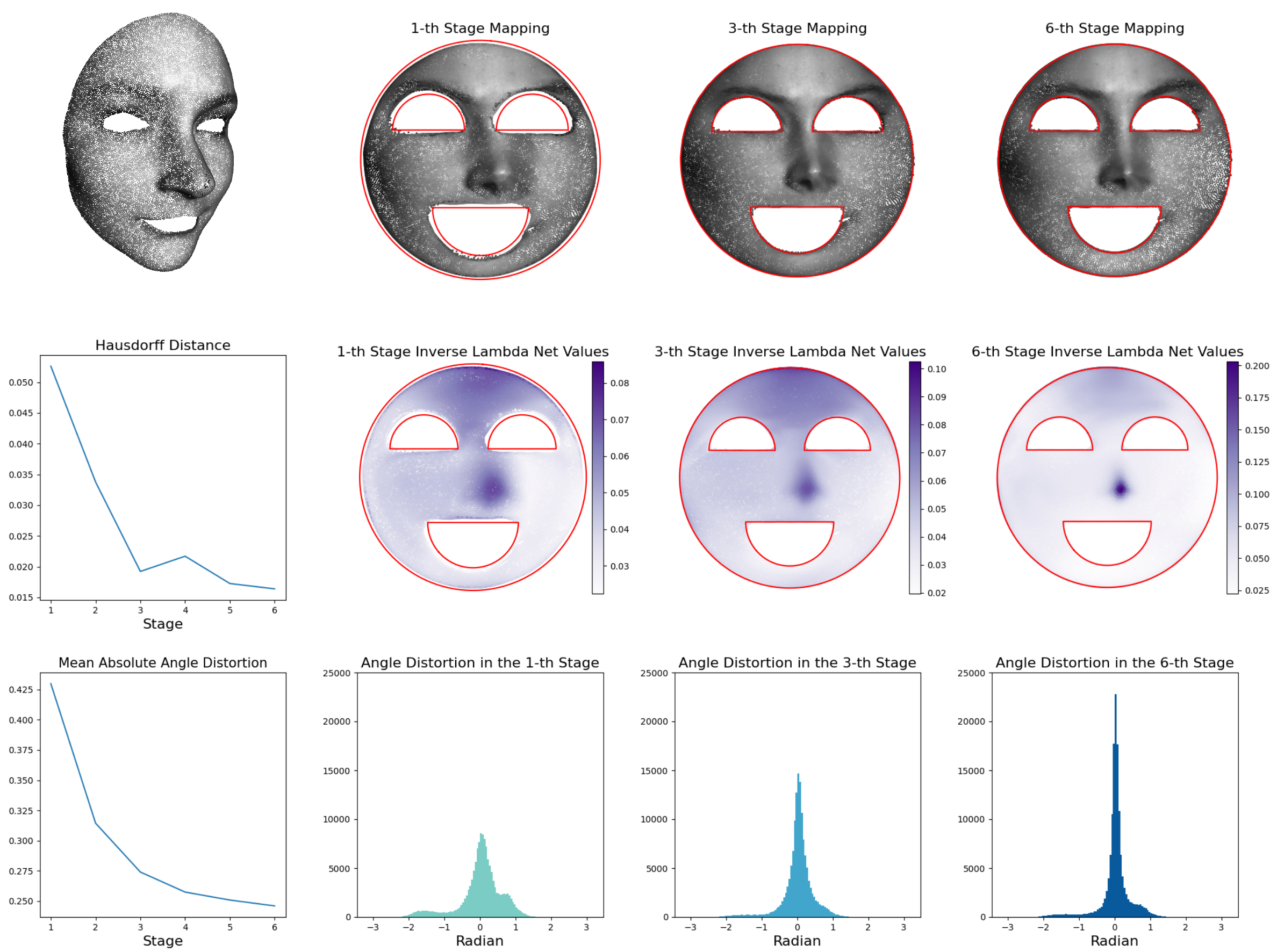}
    \caption{
        Results for parametrizing a human face, with its mouth and eyes removed, to a prescribed smiling face. 
        The organization of the figure is the same as Figure \ref{fig: fixedBdy sophie}. 
    }
    \label{fig: fixedBdy sophie noEyeMouth}
\end{figure}

Mouth and eyes are removed from the previously used human face, resulting in a point cloud with $23871$ points. 
The point are used as an input point cloud surface in computing a parametrization to the smiling face drawn in Figure \ref{fig: param domain smiling car}. The experimental results are shown below in Figure \ref{fig: fixedBdy sophie noEyeMouth}. The plots in the figure are similar to the above experiment, which parametrized the simply-connected human face to the unit disk. 
By observing the point cloud mappings and histograms for angle distortion, we can see the mismatch between the point cloud mapping and high angle distortion appeared in the earlier stages, while they are greatly improved in the latter stages. This provides evidence to the decreasing trends shown in the two line plots for Hausdorff distance and mean absolute angle distortion against stages. The changes in $\lambda_{inv}$ values also reflect the changes in the scaling of the mappings. 

\begin{figure}[!h]
    \centering
    \includegraphics[width=.8\textwidth]{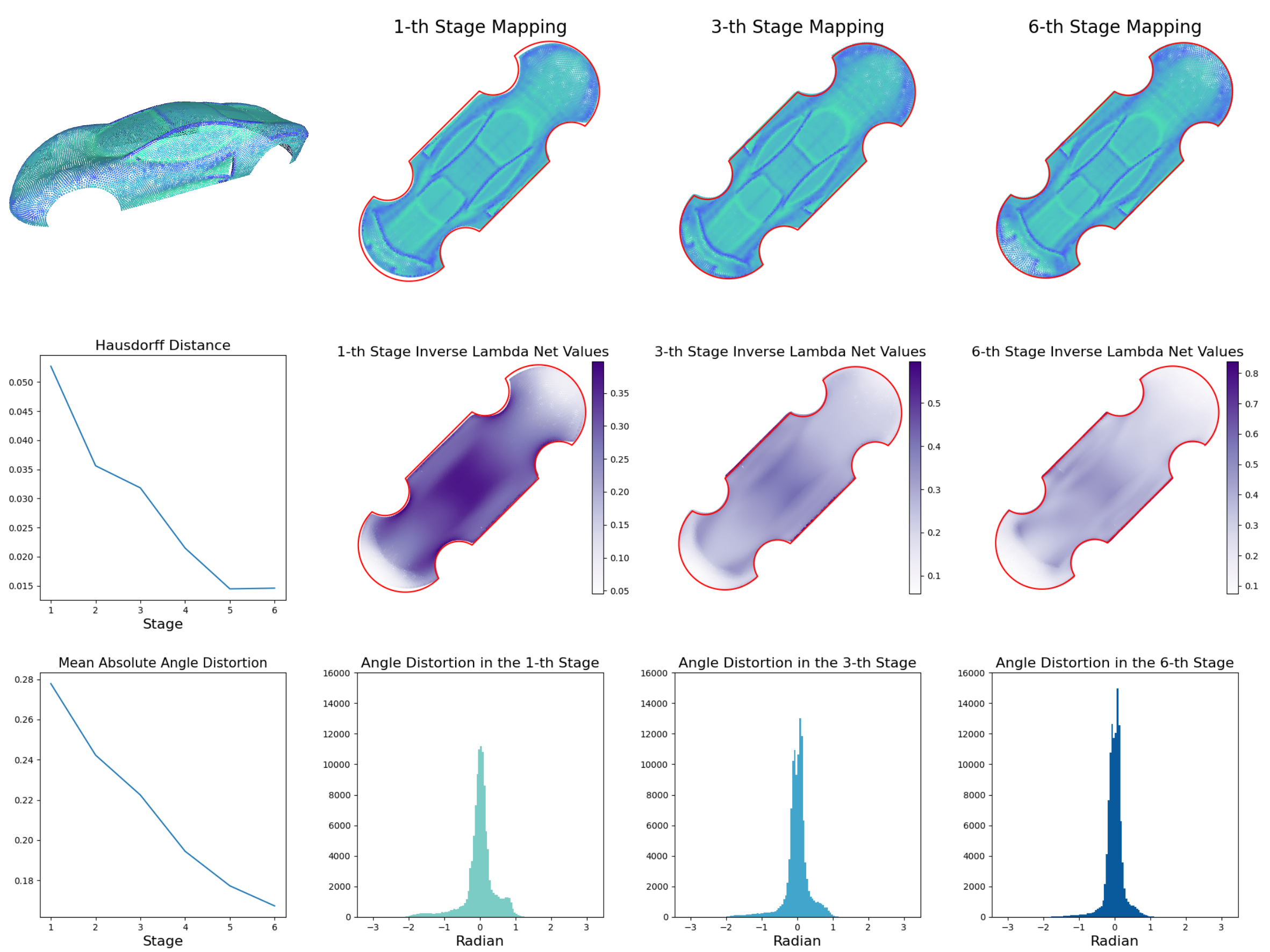}
    \caption{
        Results for parametrizing a car shell, to a prescribed car-shaped domain. 
        The organization of the figure is the same as Figure \ref{fig: fixedBdy sophie}. 
        The mappings and $\lambda_{inv}$ plots are rotated $45^\circ$.
    }
    \label{fig: fixedBdy car}
\end{figure}

In Figure \ref{fig: fixedBdy car}, the point cloud sampled from the outer shell of a car surface, drawn in the upper left corner, is used and mapped to the prescribed car-shaped parameter domain shown in Figure \ref{fig: param domain smiling car}. The other plots include Hausdorff distance, mean absolute angle distortion, and the mappings, $\lambda_{inv}$ and histograms for angle distortion. The plots for mappings and $\lambda_{inv}$ are rotated $45^\circ$ for better illustration. The overall performance are similar to the previous experiments. The improvement in preserving geometry and matching the parameter domain can be clearly observed in the plots. Alongside the previous experiment mapping a multiply-connected human face to a smiling face, shown in Figure \ref{fig: fixedBdy sophie noEyeMouth}, the effectiveness of our proposed HAND and LEG in handling point cloud surfaces and parameter domains with complicated shapes or topologies is demonstrated. 

\subsection{Domain-Landmark-Constrained Surface Parametrization}
\label{sec: DLC surf param}

\begin{figure}[!h]
    \centering
    \includegraphics[width=.25\textwidth]{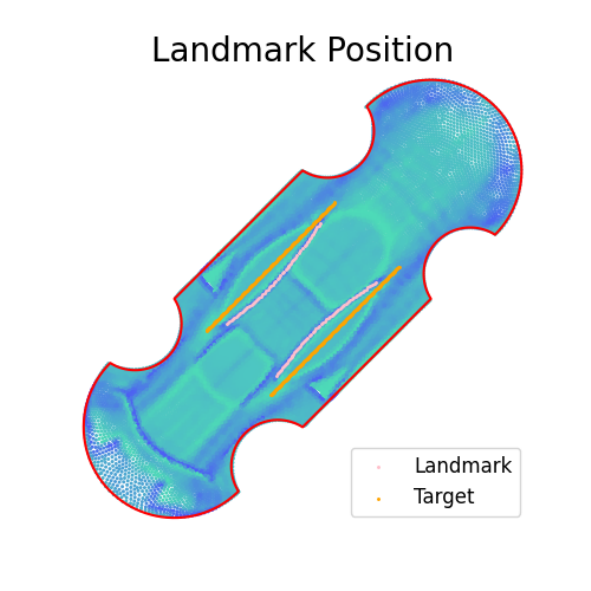}
    \caption{
        The landmarks and target position, drawing on the car shell parametrization result previously without landmark matching. The pink dots are the landmarks, while the orange dots are the target positions. The plot is rotated $45^\circ$.
    }
    \label{fig: ldmkMatch car landmark positions}
\end{figure}

The next experiment extends the car shell parametrization to include landmark matching. The landmarks and target positions are illustrated in Figure \ref{fig: ldmkMatch car landmark positions}, based on the previous parametrization result. The landmarks chosen are the points on the two upper edges of the two side windows, totally $134$ points, while the target consists of $400$ points sampled from two straight lines $ [ -0.5, 0.5 ] \times \{ -0.25, 0.25 \}$. The landmark matching is done by treating the two edges of the windows as a single landmark set, and the two straight lines as a single target set, instead of dividing them into two landmark-target pairs. 

\begin{figure}[!h]
    \centering
    \includegraphics[width=.8\textwidth]{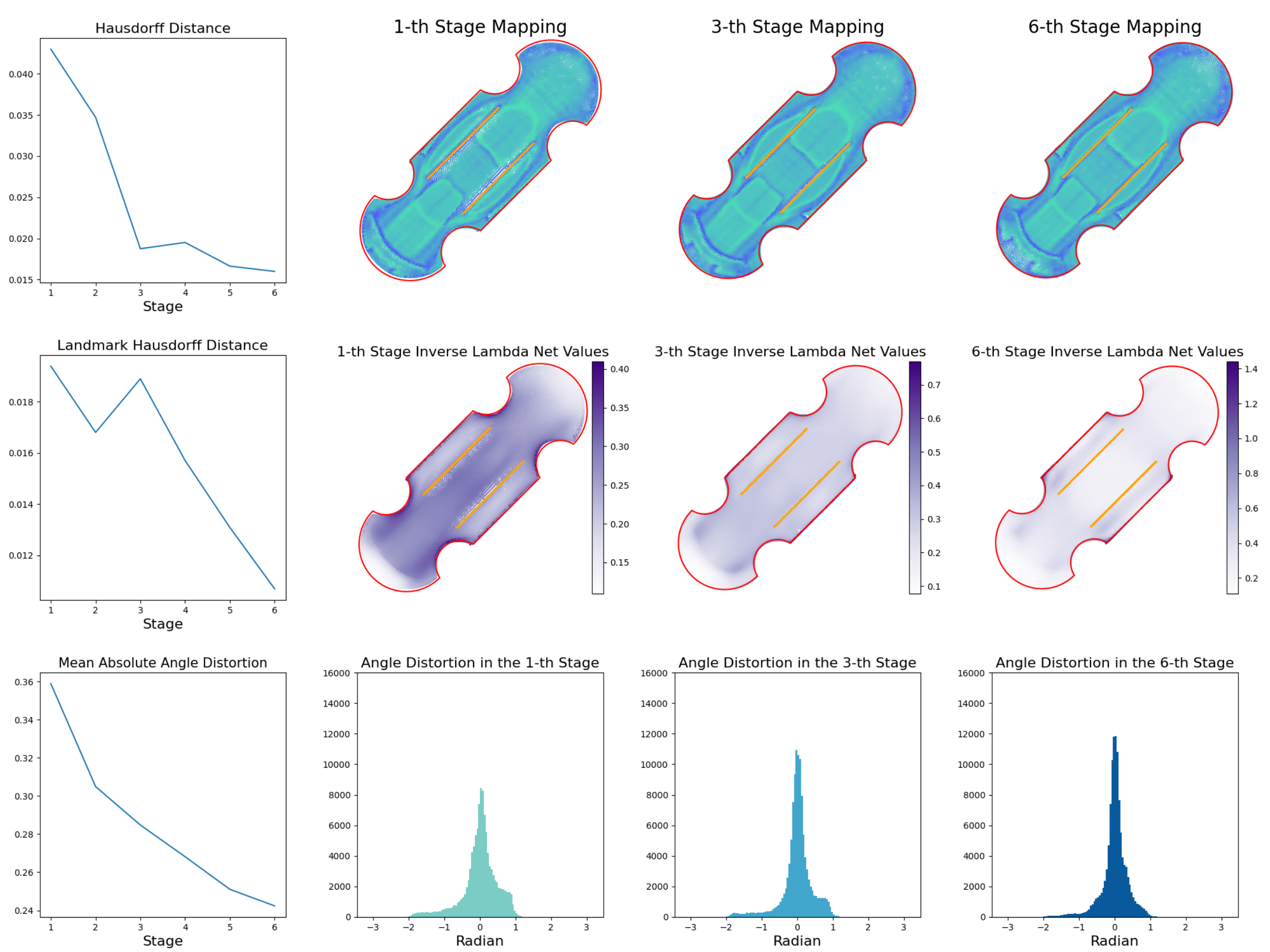}
    \caption{
        Results for parametrizing a car shell, to a prescribed car-shaped domain, together with landmark matching. 
        In the first column, line plot for Hausdorff distance for matching parameter domain, Hausdorff distance for landmark matching, mean absolute angle distortion (in radian) against stages are plotted from top to bottom. 
        The second to the last columns have the same organization with those in Figure \ref{fig: fixedBdy sophie}. 
        The pink dots are the target positions. The red lines indicate the boundary of the car-shaped domain. The mappings and $\lambda_{inv}$ plots are rotated $45^\circ$.
    }
    \label{fig: ldmkMatch car}
\end{figure}

The parametrization results with landmark matching are demonstrated in Figure \ref{fig: ldmkMatch car}. The line plots in the first column show the Hausdorff distance between point cloud mapping with parameter domain, the Hausdorff distance between landmark mapping and target position, and the mean absolute angle distortion against stages. The other three columns show the mappings, $\lambda_{inv}$ and histograms for angle distortion after $1^\text{st}$, $3^\text{rd}$, $6^\text{th}$ stages, as before. 
The improvement in matching parameter domain and preserving geometry is observed through the decrease in Hausdorff distance with parameter domain and mean absolute angle distortion.  
The landmark mappings are drawn in pink, and targets are drawn in orange. However, the pink points are barely visible, indicating that the landmark mappings and targets are well-aligned. 
Also, the line plot for landmark Hausdorff distance shows a decreasing trend, indicating an improvement in landmark matching, even though this improvement is not clearly visible in the mapping plots. Notably, our method treats the landmark edges on the point cloud surface and the landmark lines in the target parameter domain as single sets without predefined correspondences, yet it automatically identifies the correct mapping between each landmark edge and its corresponding landmark line.

%% file: sections/sec8-application.tex
\section{Applications}
\label{sec: application}
In this section, we discuss and introduce some of the applications of our proposed methods, including boundary detection and surface reconstruction. 

\subsection{Boundary Detection}
Boundary information of geometric object can be applied in academic research, like object registration and shape analysis. In industry, it can also aid in quality control and fault detection in manufacturing. 
The approach for boundary detection is performed on the simply connected human face and car shell surface. The result mappings of free-boundary parametrization, shown in Figure \ref{fig: freeBdy aio}, discussed in Section \ref{sec: FB surf param} are used. 

In our methods, the free-boundary parametrization minimizes LEG $D_\sigma(f, \lambda_{inv}, \mathcal{X})$ to obtain the mapping $f(\mathcal{X})$. Delaunay triangulation is then applied on $f(\mathcal{X})$, and a triangular mesh structure $(f(\mathcal{X}), \mathcal{T})$ is obtained. On the mesh, edges that only appear in one triangle face is detected as boundary edge. 

However, the Delaunay triangulation gives the triangulation of the convex hull of $f(\mathcal{X})$, but since the mapping is not always convex, some nonsensical triangular faces may be included. 
Suppose the all edge lengths of the ground truth surface mesh are shorter than $h > 0$. The triangular faces in $(\mathcal{X}, \mathcal{T})$ with at least one edge longer than $h$ can be defined as extra faces and removed from $\mathcal{T}$. Boundary detection is then applied to the refined triangular mesh to obtain the boundary of the point cloud. 

The result free-boundary mappings, demonstrated on Section \ref{sec: FB surf param}, of human face and car shell, shown in Figures \ref{fig: freeBdy aio}, are used, and the boundary detection results are shown in Figure \ref{fig: bdy detect AIO}. By drawing the results on the mappings in 2D plane, the figures show the ground truth boundaries for comparison, the convex hulls of the point cloud mappings, and the predicted boundaries with different choices of threshold $h$, all drawn in red lines. In Figure \ref{fig: bdy detect AIO}, the predicted boundaries using the finest threshold, $h = 0.015$ for human face and $h = 0.05$ for car shell, are drawn on the 3D point cloud surfaces. 
The results demonstrate that the predicted boundaries with an appropriate choice of $h$ is able to recover the ground truth boundaries. 

\begin{figure}[!h]
    \centering
    \includegraphics[width=\textwidth]{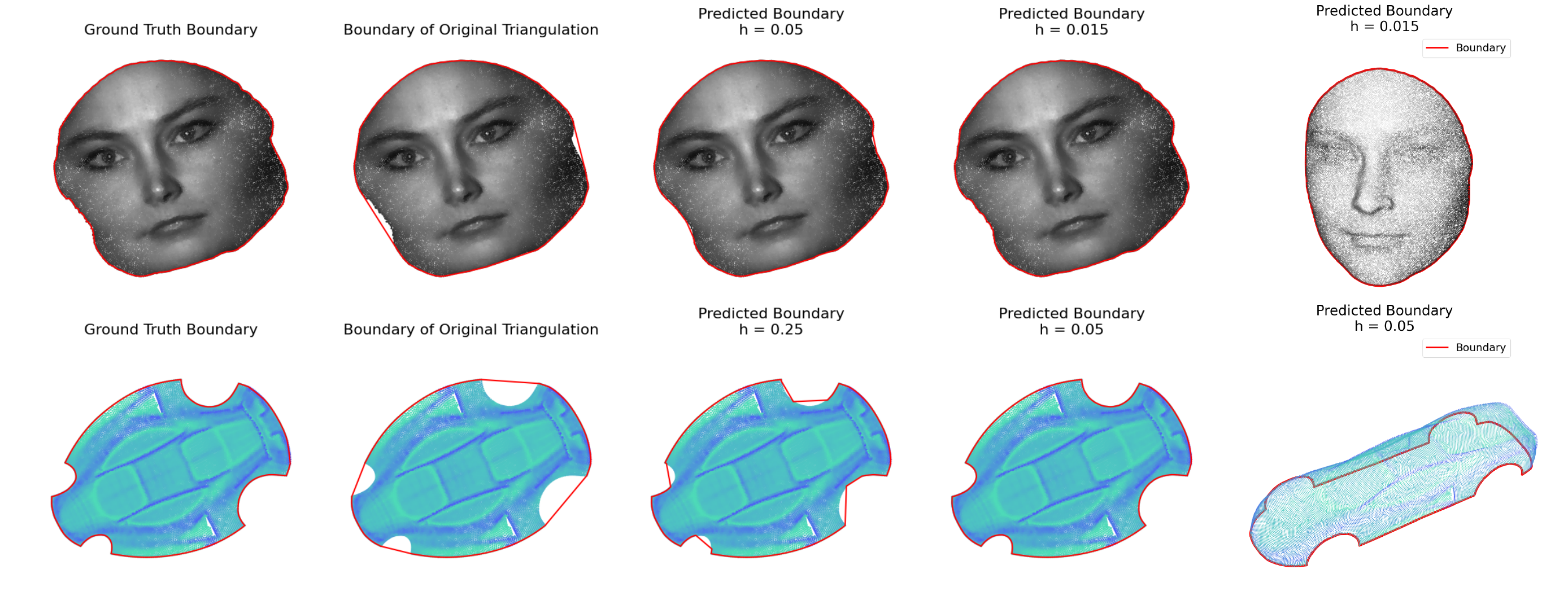}
    \caption{Result of boundary detection on a human face (top) and a car shell (bottom). Left to right: ground truth boundaries, convex hulls, and two detected boundaries with different thresholds, the detected boundaries with finer thresholds on the original point clouds. }
    \label{fig: bdy detect AIO}
\end{figure}




\subsection{Surface Reconstruction}
Surface reconstruction, which recovers geometric objects from unstructured point clouds, can be applied in finite element analysis, medical imaging, computer aided design, and various other fields. 

In Section \ref{sec: DC surf param}, point cloud surfaces are mapped to prescribed parameter domain, in where planar triangular mesh can be created. 
The planar meshes are generated by a Python library pygmsh \cite{Schlomer_pygmsh_A_Python}, which is an interface to Gmsh \cite{geuzaine2009gmsh}. 
The inverse mapping of the vertices of the newly created mesh can be obtained via interpolation using the mapping of the point cloud. A reconstructed surface triangular mesh can then be formed using pre-images of the vertices and their connectivity. 

\begin{figure}[!h]
    \centering
    \includegraphics[width=.6\textwidth]{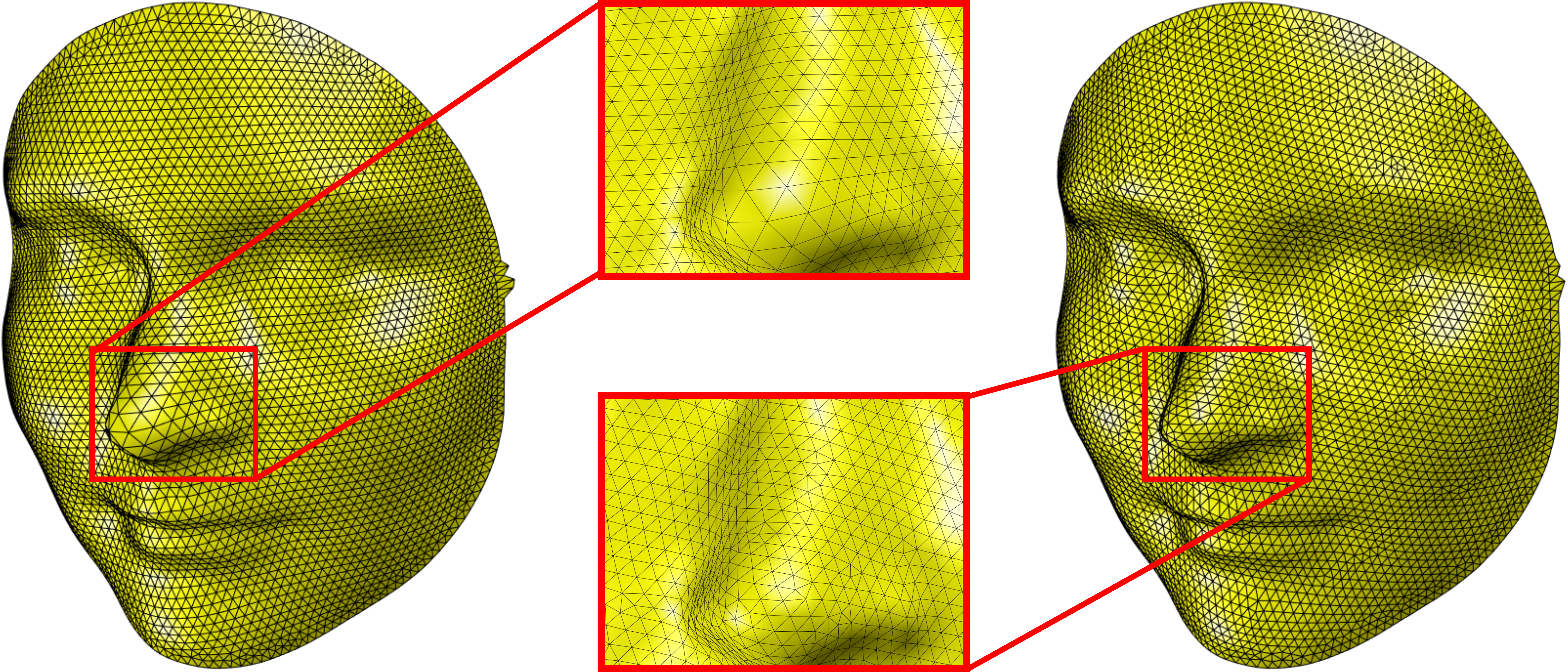}
    \caption{
        The human face triangular meshes reconstructed from uniform mesh (left) and $\lambda_{inv}-$dependent mesh (right). The close-ups to their noses are shown in the red boxes. 
    }
    \label{fig: sophie surf recon}
\end{figure}

\begin{figure}[!h]
    \centering
    \includegraphics[width=.6\textwidth]{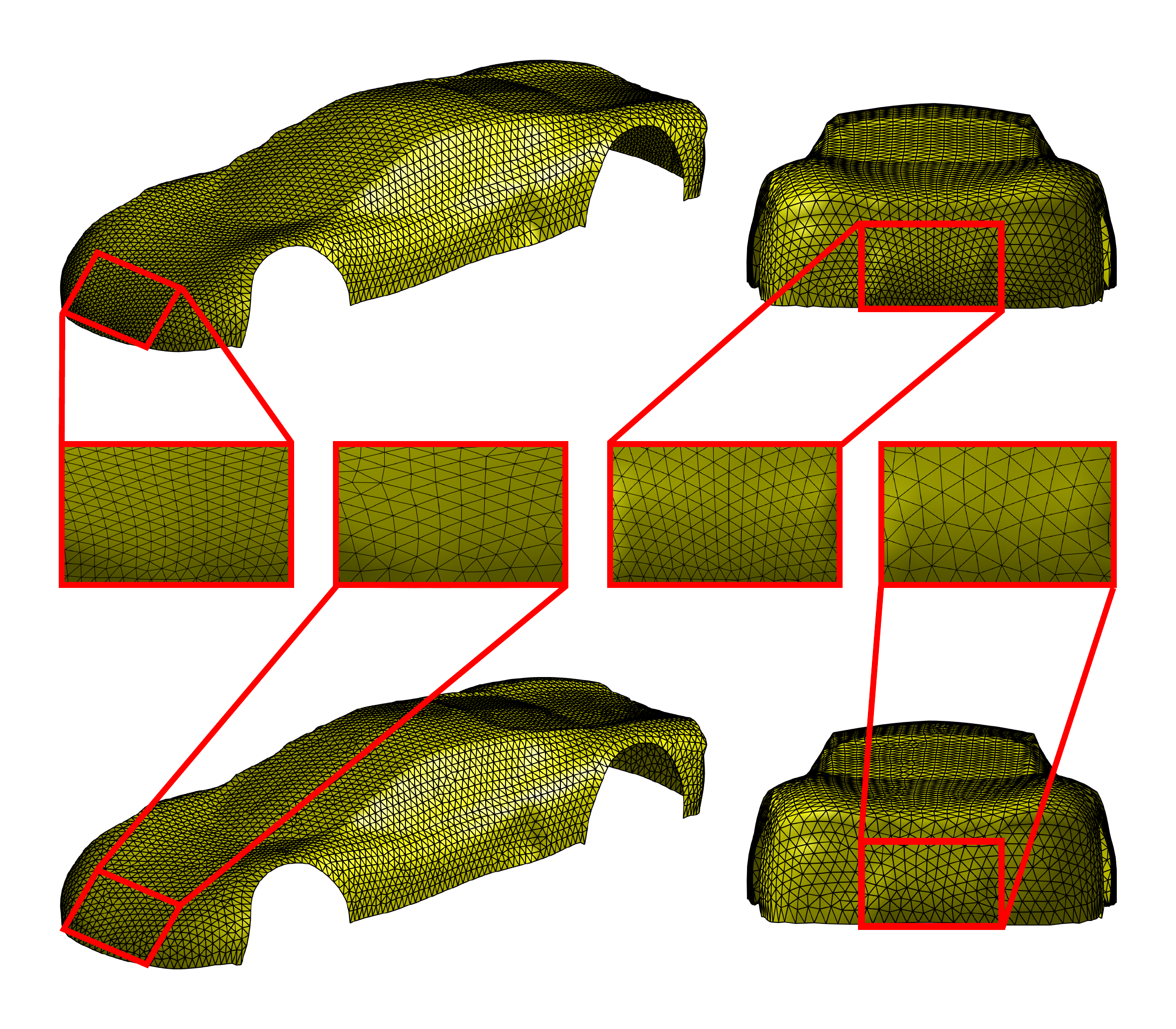}
    \caption{
        The car triangular meshes reconstructed from uniform mesh (upper) and $\lambda_{inv}-$dependent mesh (lower). 
        They are viewed from the side (left) and back (right). 
        The close-ups to local regions are shown in the red boxes. 
    }
    \label{fig: car surf recon}
\end{figure}

In creating the 2D triangular meshes on the parameter domains, we tried two methods: one was to create a uniform mesh, and the other was to create a mesh with edge lengths are depending on the values of $\lambda_{inv}$ networks. More specifically, for the latter method, the edge lengths were approximately proportional to $1/\sqrt{\lambda_{inv}}$. The idea is that the value represents the scaling of Euclidean distance between the original point cloud and its mapping, thereby creating a triangular mesh with more uniform sizes of triangular faces. 

The surface of human face in Figure \ref{fig: fixedBdy sophie} and car in Figure \ref{fig: fixedBdy car} are reconstructed with the two described methods. The reconstruction results are shown in Figures \ref{fig: sophie surf recon} and \ref{fig: car surf recon}. 
In Figure \ref{fig: sophie surf recon}, the reconstructed human faces are shown with the close-ups of the noses. The closes-ups of the noses suggesting that the $\lambda_{inv}$-depending mesh appears to be more uniform.  
In Figure \ref{fig: car surf recon}, the reconstructed surface of the car point cloud are shown from side and back. Close-ups also show that in regions that are generally smooth, the upper mesh has unnecessarily dense triangular faces, due to the scaling of mapping, while the lower mesh has a more uniform distribution. 

%% file: sections/sec9-conclusion.tex
\section{Conclusions}
\label{sec: conclusions}
In this paper, we formulate the problem of point cloud surface parametrization into an optimization problem, using only the coordinates of points. Two novel loss functions, Hausdorff Approximation from Node-wise Distances (HAND) and Localized Energy for Geometry (LEG), have been developed specifically for the purpose. 
Theoretical analysis is performed for justifying the geometric meaning of them. 
By minimizing HAND, the mapping-domain and landmark-target alignments can be improved. On the other hand, LEG is a estimation of the angle distortion between surface and mapping. 
The functions $f$ and $\lambda_{inv}$ in the optimization problem are modeled by neural networks and optimized using stochastic gradient descent.
Based on the theorems proved, an optimization algorithm for further improving the parametrization quality is also developed. 

In order to show the effectiveness of our methods, the experiments are conducted in various manners, including examining the performance of HAND, free-boundary, domain-constrained and domain-landmark-constrained parametrization. 
Despite the random initialization of parameters in the neural networks, the alignments between point cloud mappings and parameter domains, and between landmark mappings and targets, are accurate. Additionally, the geometry between the mappings and original surfaces is well preserved, suggesting the effectiveness and robustness of our methods. 
We introduce boundary detection and surface reconstruction as some applications to our methods, showing promising performance. 

Some potential future directions of our methods include acceleration in processing large datasets of point clouds, extension to closed surfaces and non-planar domains, or high dimensional mapping problems. Further controls on the geometry or alignment for more specific applications can also be investigated.